\documentclass[runningheads]{llncs}

\usepackage{booktabs} 

\usepackage[utf8]{inputenc} %
\usepackage[english]{babel}%
\usepackage[babel=true]{csquotes}%
\usepackage[T1]{fontenc} %
\usepackage{amssymb} %
\usepackage{verbatim} %
\usepackage{empheq}%
\usepackage{latexsym}%
\usepackage{amsfonts}%
\usepackage{stmaryrd} %
\usepackage{color} %
\usepackage{tikz} %
\usepackage[toc,page]{appendix}%
\usepackage{varwidth}

\usepackage{microtype}
\DisableLigatures{encoding = *, family = *}

\usepackage{enumitem} 
\setlist[enumerate]{topsep=0pt,itemsep=-1ex,partopsep=1ex,parsep=1ex}%

\usetikzlibrary{arrows,shapes,snakes,automata,backgrounds,petri}
\usetikzlibrary{positioning}%
\usepackage{hyperref} %

\usepackage{dsfont}%
\usepackage{array}%
\usepackage[linesnumbered,ruled,vlined]{algorithm2e}
\usepackage{amsmath} 
\usepackage[draft]{todonotes} %
\presetkeys{todonotes}{inline, backgroundcolor=orange!10}{}%

\usepackage{tikz} \usepackage{tikzpeople} \tikzstyle{entity} =
[rectangle, rounded corners, minimum width=2.5cm, minimum
height=1cm,text centered, draw=black!50, fill=green!20]

\usepackage{multicol}

\usepackage{graphicx} %
\DeclareGraphicsExtensions{.pdf,.png,.jpeg,.jpg,.mps,.eps,.ps} %

\usepackage{ifdraft}%
\usepackage{ifthen}%

\newboolean{bda} %
\setboolean{bda}%
{false}%
\newcommand\bda[1]{\ifthenelse{\boolean{bda}}{\ #1\ }{}}%
\newcommand\notbda[1]{\ifthenelse{\not \boolean{bda}}{\ #1\ }{}}%

\title{From Task Tuning to Task Assignment in Privacy-Preserving Crowdsourcing Platforms}

\author{Joris Duguépéroux, Tristan Allard}
\institute{Univ Rennes, CNRS, IRISA\\ Rennes, France \\ \email{\{firstname.name\}@irisa.fr}}

\begin{document}

\maketitle

\begin{abstract}
  Specialized worker profiles of crowdsourcing platforms may contain a
  large amount of identifying and possibly sensitive personal
  information (\emph{e.g.,} personal preferences, skills, available
  slots, available devices) raising strong privacy concerns. This led
  to the design of privacy-preserving crowdsourcing platforms, that
  aim at enabling efficient crowdsourcing processes while providing
  strong privacy guarantees even when the platform is not fully
  trusted. In this paper, we propose two contributions. First,
  we propose the \texttt{PKD} algorithm with
  the goal of supporting a large variety of aggregate usages of worker
  profiles within a privacy-preserving crowdsourcing platform. The
  \texttt{PKD}
  algorithm
  combines together homomorphic encryption and differential privacy
  for computing (perturbed) partitions of the multi-dimensional space
  of skills of the actual population of workers and a (perturbed)
  \texttt{COUNT} of workers per partition. Second, we propose to
  benefit from recent progresses in Private Information Retrieval
  techniques in order to design a solution to task assignment that is
  both private and affordable. We perform an in-depth study of the
  problem of using PIR techniques for proposing tasks to workers, show
  that it is NP-Hard, and come up with the \texttt{PKD PIR Packing}
  heuristic that groups tasks together according to the partitioning
  output by the \texttt{PKD} algorithm.
  In a nutshell, we design the
  \texttt{PKD} algorithm and the \texttt{PKD PIR Packing} heuristic, we prove formally their security against
  \emph{honest-but-curious} workers and/or platform, we analyze their
  complexities, and we demonstrate their quality and affordability in
  real-life scenarios through an extensive experimental evaluation
  performed over both synthetic and realistic datasets.
\end{abstract}

\section{Introduction}
\label{sec:intro}%

Crowdsourcing platforms are online intermediates between
\emph{requesters} and \emph{workers}. The former have \emph{tasks} to
propose to the latter, while the latter have \emph{profiles}
(\emph{e.g.,} skills, devices, experience, availabilities) to propose
to the former. Crowdsourcing platforms have grown in diversity,
covering application domains ranging from
micro-tasks\footnote{\url{https://www.mturk.com/}} or
home-cleaning\footnote{\url{https://www.handy.com/}} to collaborative
engineering\footnote{\url{https://www.kicklox.com/}} or specialized
software team design\footnote{\url{https://tara.ai/}}. %

The efficiency of crowdsourcing platforms especially relies on the
wealth of information available in the profiles of registered
workers. Depending on the platform, a profile may indeed contain an
arbitrary amount of information: professional or personal skills,
daily availabilities, minimum wages, diplomas, professional
experiences, centers of interest and personal preferences, devices
owned and available, \emph{etc.} %
This holds especially for platforms dedicated to specialized tasks
that require strongly qualified workers\footnote{We adopt in this
  paper a broad definition of \emph{crowdsourcing}, including in
  particular freelancing platforms (similarly to
  \cite{zebralancer}).}. But even micro-tasks platforms may maintain
detailed worker profiles (see, \emph{e.g.,} the \emph{qualification}
system of \emph{Amazon Mechanical Turk} that maintains so-called
\emph{premium
  qualifications}\footnote{\url{https://requester.mturk.com/pricing}}
- \emph{i.e.,} sociodemographic information such as \emph{age range},
\emph{gender}, \emph{employment}, \emph{marital status}, \emph{etc.}
- in the profiles of workers willing to participate to surveys). %
The availability of such detailed worker profiles is of utmost
importance to both requesters and platforms because it enables:
\begin{description}
\item [Primary usages of worker profiles:] to target the specific set of workers
  relevant for a given task (through \emph{e.g.,} elaborate task
  assignment algorithms~\cite{Mavridis:2016:UHS:2872427.2883070}).
\item [Secondary usages of worker profiles:] to describe the
  population of workers available, often through \texttt{COUNT}
  aggregates, in order \emph{e.g.,} to promote the platform by
  ensuring requesters that workers relevant for their tasks are
  registered on the platform\footnote{See for example the Kicklox
    search form (\url{https://www.kicklox.com/en/}) that inputs a list
    of keywords (typically skills) and displays the corresponding
    number of workers available.}, or to participate to the task
  design by letting requesters fine-tune the tasks according to the
  actual population of workers (\emph{e.g.,} setting wages according
  to the rarity of the skills required, adapting slightly the
  requirements according to the skills available). %
\end{description}
Both primary and secondary usages are complementary and usually
supported by today's crowdsourcing platforms, in particular by
platforms dedicated to highly skilled tasks and workers\footnote{See
  for example, Kicklox (\url{https://www.kicklox.com/en/}) or Tara
  (\url{https://tara.ai/}). The secondary usage consisting in
  promoting the platform is sometimes performed through a public
  access to detailed parts of worker profiles (\emph{e.g.,} Malt
  (\url{https://www.malt.com/}), 404works
  (\url{https://www.404works.com/en/freelancers})).}.



However, the downside of fine-grained worker profiles is that detailed
information related to personal skills can be highly
identifying %
(\emph{e.g.,} typically a unique combination of
location/skills/centers of interest) or sensitive (\emph{e.g.,} costly
devices or high minimum wages may correspond to a wealthy individual,
various personal traits may be inferred from centers of
interest\footnote{See, \emph{e.g.,}
  \url{http://applymagicsauce.com/about-us}}). %
Recent privacy scandals have shown that crowdsourcing platforms are
not immune to negligences or misbehaviours. %
Well-known examples include cases where personally identifiable
information of workers is trivially exposed
online~\cite{LeaseHullmanBighamEtAl2013} or where precise geolocations
are illegitimately accessed\footnote{For example, internal emails that
  were leaked from Deliveroo indicate that the geolocation system of
  Deliveroo was used internally for identifying the riders that
  participated to strikes against the platform.\\
  \url{https://www.lemonde.fr/culture/article/2019/09/24/television-cash-investigation-a-la-rencontre-des-nouveaux-proletaires-du-web_6012758_3246.html}}$^,$\footnote{In
  another example, an Uber executive claimed having tracked a
  journalist using the company geolocation system. \\
  \url{https://tinyurl.com/y4cdvw45}%
}. %
It is noticeable that workers nevertheless expect platforms to secure
their data and to protect their privacy in order to lower the privacy
threats they face~\cite{crowdsurvey}. Moreover, in a legal context
where laws firmly require businesses and public organizations to
safeguard the privacy of individuals (such as the European
GDPR\footnote{\url{https://eur-lex.europa.eu/eli/reg/2016/679/oj}} or
the California Consumer Privacy
Act\footnote{\url{https://www.caprivacy.org/}}), legal compliance is
also a strong incentive for platforms for designing and
implementing sound privacy-preserving crowdsourcing processes. Ethics
in general, and privacy in particular, are indeed clearly identified
as key issues for next generation \emph{future of work}
platforms~\cite{fow}. Most related privacy-preserving works have
focused on the primary usage of worker profiles, \emph{i.e.,} the
task-assignment problem 
(\emph{e.g.,} based on additively-homomorphic
encryption~\cite{kajino2015phd} or on local differential
privacy~\cite{beziaud}\footnote{Note that limiting the information
  disclosed to the platform (\emph{i.e.,} perturbed information about
  worker profiles) relieves platforms from the costly task of
  handling personal data. The European GDPR indeed explicitely
  excludes anonymized data from its scope (see Article 4, Recital 26
  \url{https://gdpr-info.eu/recitals/no-26/}).}).  %

Our goal in this paper is twofold: (1) \emph{consider both primary and
  secondary usages as first-class citizens} by proposing a
privacy-preserving solution for computing multi-dimensional
\texttt{COUNTs} over worker profiles and a task assignment algorithm
based on recent affordable Private Information Retrieval (PIR)
techniques, and (2) \emph{integrate well with other privacy-preserving
  algorithms possibly executed by a platform} without jeopardizing the
privacy guarantees by requiring our privacy model to be composable
both with usual computational cryptographic guarantees provided by
real-life encryption schemes and with classical differential privacy
guarantees as well. %



%
%
%
%



The two problems are not trivial. First, we focus on secondary
usages. %
The problem of computing multi-dimensional \texttt{COUNTs} over
distributed worker profiles in a privacy-preserving manner is not
trivial. Interactive approaches - that issue a privacy-preserving
\texttt{COUNT} query over the set of workers each time needed
(\emph{e.g.,} a requester estimates the number of workers qualified
for a given task) - are inadequate because the number of queries would
be unbounded. This would lead to out-of-control information disclosure
through the sequence of \texttt{COUNTs}
computed~\cite{DBLP:journals/corr/abs-1810-05692,DBLP:conf/pods/DinurN03}. Non-interactive
approaches are a promising avenue because they compute, once for all
and in a privacy-preserving manner, the static data structure(s) which
are then exported and queried by the untrusted parties (\emph{e.g.,}
platform, requesters) without any limit on the number of queries. More
precisely, on the one hand, hierarchies of histograms are well-known
data structures that support \texttt{COUNT} queries and that cope well
with the stringent privacy guarantees of differential
privacy~\cite{Qardaji:2013:UHM:2556549.2556576}. However, they do not
cope well with more than a few
dimensions~\cite{Qardaji:2013:UHM:2556549.2556576}, whereas a worker
profile may contain more skills (\emph{e.g.,} a dozen), and they
require a trusted centralized platform. On the other hand,
privacy-preserving spatial
decompositions~\cite{psd,Zhang:2016:PDP:2882903.2882928}
are more tolerant to a higher number of dimensions but require as well
a trusted centralized platform. Second, algorithms for assigning tasks
to workers while providing sound privacy guarantees have been proposed
as alternatives against naive \emph{spamming approaches} - where all
tasks are sent to all workers. However they are either based (1) on
perturbation only (e.g., \cite{beziaud}) and suffer from a severe drop
in quality or (2) on encryption only (e.g., \cite{kajino2015phd}) but
they do not reach realistic performances, or (3) they focus on the
specific context of spatial crowdsourcing and geolocation data (e.g.,
\cite{to2018privacy}).

\begin{figure}[htbp]
  \centering
  \resizebox{1\linewidth}{!}{\begin{tikzpicture}[every text node part/.style={align=center}]


\node[builder,minimum size=.5cm] (W1) at (-3.9,.7) {};
\node[builder,minimum size=.5cm] (W2) at (-2.7,.7) {Workers};
\node[builder,minimum size=.5cm] (W3) at (-1.5,.7) {};

\node[police,minimum size=.5cm] (PF) at (-2.7,3.5) {Platform};
\node[] (PF_pos_1) at (-3,2.9) {};
\node [] (PF_pos_2) at (-2.7,2.9) {};
\node [] (PF_pos_3) at (-2.4,2.9) {};

\draw [->] (W1) to [bend left=15] node [bend right] {} (PF_pos_1); %
\draw [<-] (W1) to [bend right=15] node [bend right] {} (PF_pos_1); %
\draw [->] (W2) to [bend left=15] node [bend right] {} (PF_pos_2); %
\draw [<-] (W2) to [bend right=15] node [bend right] {} (PF_pos_2); %
\draw [->] (W3) to [bend left=15] node [bend right] {} (PF_pos_3); %
\draw [<-] (W3) to [bend right=15] node [bend right] {} (PF_pos_3); %

\draw [fill=black, white] (-4,1.6) rectangle (0,2.4) node[pos=.5] {The PKD Algorithm};
\node [] (PKD) at (-2.8,2) {Privacy preserving\\exchanges};

\node [] (PKD) at (-2.7,-1.1) {\textbf{Computation:} \\ \textbf{P}rivacy Preserving \\ \textbf{KD} Tree Algorithm};

\path [->, very thick] (-1.4,1.7) edge (-0.6,1.7); %


\node[police,minimum size=.5cm] (PF) at (1.5,3.8) {Platform};

\node [] (origin) at (-0.2,-0.2) {$0$};
\node [] (Sk1_1) at (3.2,-0.2) {$1$};
\node [] (Sk2_1) at (-0.2,3.2) {$1$};
\node [] (Sk1) at (1.5,-0.2) {\emph{$C++$}};
\node [rotate=90] (Sk1) at (-0.2,1.5) {\emph{Network}};

\draw (0,0) rectangle (3,3) node[pos=.5] {};
\path [-] (0,1.5) edge (3,1.5);
	\path [-] (1,1.5) edge (1,3);
		\path [-] (0,2.5) edge (1,2.5);
		\path [-] (1,2.1) edge (3,2.1);
	\path [-] (2,0) edge (2,1.5);
		\path [-] (0,1.1) edge (2,1.1);
		\path [-] (2,0.4) edge (3,0.4);

\node [] (1) at (1,0.5) {$\sim46$};
\node [] (2) at (1,1.25) {$\sim53$};
\node [] (3) at (0.5,2) {$\sim54$};
\node [] (4) at (0.5,2.75) {$\sim49$};
\node [] (5) at (2.5,0.9) {$\sim50$};
\node [] (6) at (2.5,0.2) {$\sim45$};
\node [] (7) at (2,1.75) {$\sim58$};
\node [] (8) at (2,2.5) {$\sim46$};

\node [] (heatmap) at (1.5,-1.1) {\textbf{Output:}\\\textbf{Perturbed} partitions and \\ \texttt{COUNT}s of workers};


\path [->, very thick] (3.2,1.7) edge (4,1.7); %

\node[police,minimum size=.5cm] (PF) at (4.7,3.8) {Platform};
\node[businessman,minimum size=.5cm] (PF) at (6.1,3.85) {Requesters};

\node at (5.2, 1.5) {
    {\begin{varwidth}{\linewidth}\begin{itemize}
        \item Helps requesters \\designing tasks
        \item Helps promoting \\the platform
        \item Helps tuning \\wages
		\item etc.
    \end{itemize}\end{varwidth}}
};

\node [] (heatmap) at (5.2,-1.1) {\textbf{Uses:}\\ Assist Platform \\and Requesters};

\end{tikzpicture}}
  \caption{Overview of the \texttt{PKD} algorithm: supporting
    secondary usages of worker profiles with privacy guarantees}
  \label{fig:crowdpriv}
\end{figure}
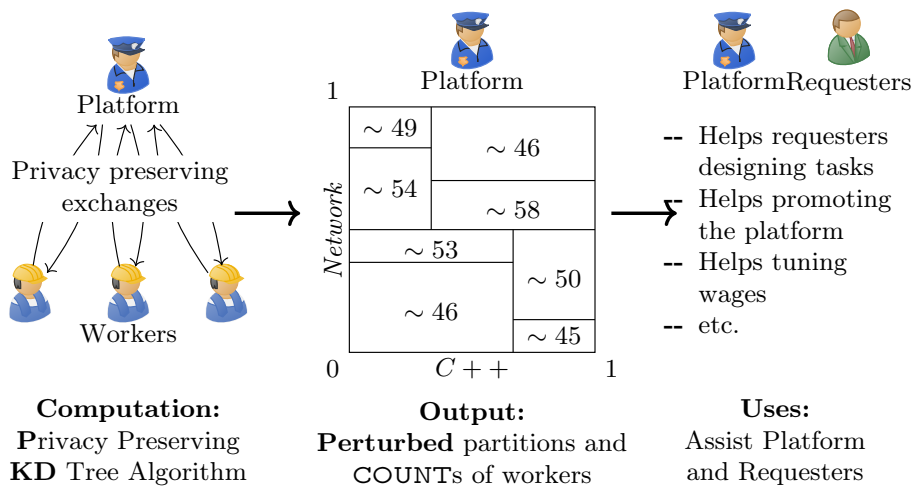

\paragraph{Our Contribution.}
First, we propose to benefit from the best of the two non-interactive
approaches described above by computing a privacy-preserving space
partitioning of the worker profiles (for coping with their
dimensionality) based on perturbed 1-dimensional histograms (for their
nice tolerance to differentially private perturbations). We propose
the \texttt{P}rivacy-preserving \texttt{KD}-Tree algorithm
(\texttt{PKD} for short, depicted in Figure~\ref{fig:crowdpriv}), a
privacy-preserving algorithm for computing a (perturbed)
multi-dimensional distribution of skills of the actual population of
workers. The \texttt{PKD} algorithm is distributed between mutually
distrustful workers and an untrusted platform. It consists in
splitting recursively the space of skills in two around the median
(similarly to the \emph{KD-tree} construction algorithm) based on the
1-dimensional histogram of the dimension being split, and it protects
workers' profiles all along the computation by combining
additively-homomorphic encryption together with differentially private
perturbation. No raw worker profile is ever communicated, neither to
the platform nor to other workers. The output of the \texttt{PKD}
algorithm is a hierarchical partitioning of the space of skills
together with the (perturbed) \texttt{COUNT} of workers per
partition (see Figure~\ref{fig:crowdpriv}). 
The \texttt{PKD} algorithm is complementary to privacy-preserving task
assignment works and can be used in conjunction with them provided
that the privacy models compose well. In particular, since our privacy
model is a computational variant of differential privacy, the
\texttt{PKD} algorithm composes well with state-of-the-art approaches
\cite{kajino2015phd,beziaud} since they are based on usual
computational cryptographic model or differential privacy model.
Second, we propose to benefit from recent progresses in Private
Information Retrieval techniques~\cite{xpir} in order to design a
solution to task assignment that is both private and affordable. We
perform an in-depth study of the problem of using PIR techniques for
proposing tasks to workers, show that it is NP-Hard, and come up with
the \texttt{PKD PIR Packing} heuristic that groups tasks together
according to the partitioning output by the \texttt{PKD}
algorithm. Obviously, the \texttt{PKD PIR Packing} heuristic composes
well with the \texttt{PKD} algorithm.

More precisely, we make the following contributions:
\begin{enumerate}
\item We design the \texttt{PKD} algorithm, a distributed
  privacy-preserving algorithm for computing a multi-dimensional
  hierarchical partitioning of the space of skills within a population
  of workers.
\item We formally prove the security of the \texttt{PKD} algorithm
  against \emph{honest-but-curious} attackers. The \texttt{PKD}
  algorithm is shown to satisfy a computational variant of
  differential privacy called the \texttt{$\epsilon_\kappa$-SIM-CDP}
  model. We provide a theoretical analysis of its complexity.
\item We provide an in-depth study of the problem of using PIR
  techniques for proposing tasks to workers and design the \texttt{PKD
    PIR Packing} heuristic that benefits from the partitioning
  computed by the \texttt{PKD} algorithm for grouping tasks
  together. We show that the \texttt{PKD PIR Packing} heuristic
  satisfies our privacy model.

\item We provide an extensive experimental evaluation of the
  \texttt{PKD} algorithm and of the \texttt{PKD PIR Packing} heuristic
  over synthetic and realistic data that demonstrates their quality
  and performance in various scenarios. Our realistic skills dataset
  is built from data dumps of \emph{StackExchange} online forums.
\end{enumerate}

The paper is organized as follows. Section~\ref{sec:problem}
introduces the participant model, the security and privacy models, and
the technical tools necessary in the rest of the
paper. Section~\ref{sec:pkd} describes the \texttt{PKD} algorithm in
details and formally analyzes its cost and
security. Section~\ref{sec:delivery} studies the problem of using PIR
techniques for task assignment, describes the \texttt{PKD PIR Packing}
heuristic, and formally analyzes its security. We discuss some details on how to allow updates for both the \texttt{PKD} algorithm and the \texttt{PKD PIR Packing} heuristic in Section~\ref{sec:discussion}. Section~\ref{sec:expe}
experimentally validates their quality and efficiency. In
Section~\ref{sec:rel}, we survey the related work. Finally,
Section~\ref{sec:conc} concludes and discusses interesting future
works.

%

%

\section{Preliminaries}%
\label{sec:problem}


\subsection{Participants Model}

Three types of participants collaborate together during our
crowdsourcing process. \emph{Workers} are interested in solving tasks
that are relevant to their profiles; \emph{requesters} propose tasks
to be solved by relevant workers; and an intermediary, \emph{i.e.,}
the \emph{platform}. %

A worker profile $p_i \in \mathcal{P}$ is represented by an
$n$-dimensional vector of floats, where each float value
$p_i[j] \in [0, 1]$ represents the degree of competency of the worker $i$
with respect to the $j^{th}$ skill. The set of skills available and
their indexes within workers' profiles is static and identical for all
profiles. %

A task $t_k \in \mathcal{T}$ is made of two parts. First, the
\emph{metadata} part is a precise description of the worker profiles
that are needed for the task completion. More precisely it is an
$n$-dimensional subspace of the space of skills. This work does not
put any constraint on the kind of subspace described in the metadata
part (\emph{e.g.,} hyper-rectangles, hyper-spheres, arbitrary set
operators between subspaces). However, for the sake of concreteness,
we focus below on metadata expressed as hyper-rectangles. More
formally, the metadata $m_k \in \mathcal{M}$ of a task
$t_k \in \mathcal{T}$ is an $n$-dimensional vector of ranges over
skills where the logical connector between any pair of ranges is the
conjunction. We call $m_k[j]$ the range of float values (between $0$ and $1$) for task $k$ and skill $j$.
The second part of a task consists in the necessary
information for performing the task and is represented as an arbitrary
bitstring $\{0, 1\}^*$. In this work, we essentially focus on the
metadata part of tasks. %
We say that a worker and a task match if the point described by the
worker profile belongs to the subspace described by the task metadata,
\emph{i.e.,} worker $p_i$ and task $t_k$ match if and only if
$\forall j \in [0, n-1]$, then $p_i[j] \in m_k[j]$.

We do not make strong assumptions on the resources offered by
participants. Workers, requesters, and the platform are equipped with
today's commodity hardware (\emph{i.e.,} the typical
CPU/bandwidth/storage resources of a personal computer). However, we
expect the platform to be available 24/7 - contrary to workers or
requesters - similarly to a traditional client/server setting.

We assume that all participants follow the \emph{honest-but-curious}
attack model in that they do not deviate from the protocol but make
use of the information disclosed, in any computationally-feasible way,
for inferring personal data. Workers may collude together up to a
reasonable bound denoted by $\tau$ in the following. %




\subsection{Privacy Tools}\label{subsec:sec}
%
\paragraph{Computational Differential Privacy.}

Our proposal builds on two families of protection mechanisms: a
\emph{differentially private} perturbation scheme and a
\emph{semantically secure} encryption scheme. The resulting overall
privacy model thus integrates the two families of guarantees
together. %
%
%
The original $\epsilon$-differential privacy model~\cite{Dwork2006}
(Definition~\ref{def:dp}) applies to a randomized function
$\mathtt{f}$ and aims at hiding the impact of any possible individual
value on the possible outputs of $\mathtt{f}$. In our context, the
function $\mathtt{f}$ is the \texttt{PKD}
algorithm. The $\epsilon$-differential privacy model requires that the
probability that any worker profile $p_i \in \mathcal{P}$ participates
to the computation of $\mathtt{f}$ be close to the probability that
$p_i$ does not participate by an $e^\epsilon$ factor, whatever the
output of $\mathtt{f}$. $\epsilon$-differential privacy holds against
information-theoretic adversaries (unlimited computational power).

\begin{definition}[$\epsilon$-differential privacy \cite{Dwork2006}]%
  \label{def:dp}%
  The randomized function $\mathtt{f}$ satisfies
  $\epsilon$-differential privacy, where $\epsilon > 0$, if: %
  $$\mathtt{Pr} [ \mathtt{f} ( \mathcal{P}_1 ) = \mathcal{O} ]
  \leq e^\epsilon \cdot \mathtt{Pr} [ \mathtt{f} ( \mathcal{P}_2 ) =
  \mathcal{O} ]$$ %
  for any set $\mathcal{O}\in Range(\mathtt{f})$ and any set of worker
  profiles $\mathcal{P}_1$ and $\mathcal{P}_2$ that differ in at most
  one profile.
\end{definition}%

The \texttt{$\epsilon_\kappa$-SIM-CDP} differential privacy
relaxation~\cite{mironov_computational_2009} requires that the
function actually computed be \emph{computationally indistinguishable}
from a pure (information theoretic) $\epsilon$-differentially private
function to adversaries whose size is polynomial in the security
parameter $\kappa \in \mathbb{N}$. The \texttt{$\epsilon_\kappa$-SIM-CDP} model
is especially relevant in our context because we combine a
differentially private perturbation scheme (information theoretic
guarantees) together with an additively-homomorphic encryption scheme
that provides computational security guarantees for performance
reasons.

\begin{definition}[$\epsilon_\kappa$-SIM-CDP privacy
  \cite{mironov_computational_2009} (simplified)]%
  \label{def:simcdp}%
  The randomized function $\mathtt{f_\kappa}$ provides
  $\epsilon_\kappa$\texttt{-SIM-CDP} if there exists a function
  $\mathtt{F_\kappa}$ that satisfies $\epsilon$-differential privacy
  and a negligible function $negl(\cdot)$, such that for every set of
  worker profiles $\mathcal{P}$, every probabilistic polynomial time
  adversary $\mathtt{A}_\kappa$, every auxiliary background knowledge
  $\zeta_\kappa \in \{0, 1\}^*$, it holds that: %
  $$|\mathtt{Pr} [ \mathtt{A}_k ( \mathtt{f}_\kappa ( \mathcal{P}, \zeta_\kappa ) ) = 1 ] -
  \mathtt{Pr} [ \mathtt{A}_k ( \mathtt{F}_\kappa ( \mathcal{P},
  \zeta_\kappa ) ) = 1 ] | \leq negl(\kappa)$$
\end{definition}

\paragraph{Achieving Differential Privacy with the Geometric Mechanism.}





A common mechanism for satisfying $\epsilon$-differential privacy with
functions that output floats or integers consists in adding random
noise to their outputs. In particular, the Geometric Mechanism (Definition~\ref{def:geom}) allows functions that output integers to be perturbed and to satisfy $\epsilon$-differential privacy, while maximizing utility for count queries as shown in~\cite{ghosh2012universally}.

\begin{definition}[Geometric mechanism~\cite{ghosh2012universally}]%
  \label{def:geom}%
  Let $\mathcal{G}$ denote a random variable following a two-sided
  geometric distribution, meaning that its probability density
  function is $g(z,\alpha) = \frac{1-\alpha}{1+\alpha} \alpha^{|z|}$
  for $z \in \mathds{Z}$.  Given any function
  $\mathtt{f} : \mathds{N}^{|\mathcal{X}|} \rightarrow \mathds{Z}^k$
  the Geometric Mechanism is defined as
  $M_G(x, f(.), \alpha) = \mathtt{f}(x) + (Y_1,…,Y_k)$ where $Y_i$ are
  independent identically distributed random variables drawn from
  $\mathcal{G}(e^{-\epsilon/\Delta \mathtt{f}})$, and
  $\Delta \mathtt{f}$ is its global sensivity
  $ \Delta \mathtt{f} = \displaystyle \max_{\mathcal{P}_1,
    \mathcal{P}_2}
  ||\mathtt{f}(\mathcal{P}_1)-\mathtt{f}(\mathcal{P}_2)||_1$ for all
  $(\mathcal{P}_1, \mathcal{P}_2)$ pairs of sets of worker profiles
  s.t. $\mathcal{P}_2$ is $\mathcal{P}_1$ with one profile more.
\end{definition}%





Intuitively, a distribution is said to be infinitely divisible if it
can be decomposed as a sum of an arbitrary number of independent
identically distributed random variables. This property allows to
distribute the generation of the noise over a set of participants. It
is valuable in contexts such as ours where no single trusted party, in
charge of generating the noise, exists. Definition~\ref{def:infinite}
below formalizes the infinite divisibility property, and
Theorem~\ref{the:inf2geom} shows that the two-sided geometric
distribution is infinitely divisible.

\begin{definition}[Infinite Divisibility]%
  \label{def:infinite}%
  A probability distribution with characteristic function $\psi$ is
  infinitely divisible if, for any integer $n\geq 1$, we have $\psi =
  \phi_n^n$ where $\phi_n$ is another characteristic function. In
  other words, a random variable $Y$ with characteristic function
  $\psi$ has the representation $Y \overset{d}{=} \sum_{i=1}^{n}X_i$
  for some independent identically distributed random variables $X_i$.
\end{definition}

\begin{theorem}[Two-sided Geometric Distribution is Infinitely
  Divisible]%
  \label{the:inf2geom}
  Let $Y$ follow two-sided geometric distribution with probability
  density function $d(z,\epsilon) = \frac{1-\epsilon}{1+\epsilon}
  \epsilon^{|z|}$ for any integer $z$. Then the distribution of $Y$ is
  infinitely divisible.  Furthermore, for every integer $n\geq 1$,
  representation of definition~\ref{def:infinite} holds. Each $X_i$ is
  distributed as $X_{1n} - X_{2n}$ where $X_{1n}$ and $X_{2n}$ are
  independent identically distributed random variable with negative
  binomial distribution, with probability density function $g(k,n)=
  \binom{k-1+1/n}{k}(1-\alpha)^k*\alpha^{1/n}$.
\end{theorem}

To prove this result, we will use a similar result for the geometric
distribution.

\begin{theorem}[Geometric Distribution is Infinitely
  Divisible~\cite{steutel2003infinite}]%
  \label{the:inf_geom}%
  Let $Y$ have a geometric distribution with probability density
  function $f(k,\alpha) = (1-\alpha)*\alpha^k$ for $k \in \mathds{N}$.
  Then the distribution of $Y$ is infinitely divisible.  Furthermore,
  for every integer $n\geq 1$, representation of~\ref{def:infinite}
  holds.  Each $X_i$ is distributed as $X_{n}$ where $X_{n}$ are
  independent identically distributed random variable with negative
  binomial distribution, with probability density function
  $g(k,n,\alpha)= \binom{k-1+1/n}{k}(1-\alpha)^k\alpha^{1/n}$.
\end{theorem}

\begin{proof}
  Using this theorem, proving that a two-sided
  geometric distribution with density function $d(z,\alpha) =
  \frac{1-\alpha}{1+\alpha} \alpha^{|z|}$is equal to the difference
  between two independent identically distributed geometric
  distributions with density function $f(k, \alpha) = (1-\alpha)
  \alpha^k$ is enough to deduce the result. Let $X_+$ and $X_-$ be two
  such random variables.

$$
P(X_+ - X_- = z) = \\ \left\{
  \begin{array}{ll}
	\mbox{if }z \geq 0\\
    \sum_{j=0}^{\infty} ((1- \alpha) \alpha^{z+j})((1- \alpha ) \alpha^j)\\
	\mbox{if } z < 0 \\
    \sum_{j=0}^{\infty} ((1- \alpha) \alpha^j)((1- \alpha ) \alpha^{-z+j})
  \end{array}
\right.
$$

$$
\begin{array}{ll}
  P(X_+ - X_- = z)&= \sum_{j=0}^{\infty}(1-\alpha)^2 \alpha^{|z| + 2j} \\
  &= (1-\alpha)^2 \alpha^{|z|} \sum_{j=0}^{\infty}(\alpha^2)^j\\
  &= (1-\alpha)^2 \alpha^{|z|} \frac{1}{1-\alpha^2}\\
  &= \frac{(1-\alpha)^2}{(1-\alpha)(1+\alpha)} \alpha^{|z|} \\
  &= \frac{1-\alpha}{1+\alpha} \alpha^{|z|} \\
  &= P(Y=z)
\end{array}
$$
for $Y$ a random variable with a two-sided geometric distribution with
parameter $\alpha$.

\end{proof}

Finally, Theorem~\ref{the:composition} 
states that differential privacy composes with itself gracefuly.

\begin{theorem}[Sequential and Parallel
  Composability~\cite{dwork2014algorithmic}]%
  \label{the:composition}%
  Let $\mathtt{f}_i$ be a set of functions such that each provides
  $\epsilon_i$-differential privacy. First, the \emph{sequential
    composability} property of differential privacy states that
  computing all functions on the same dataset results in satisfying
  $(\sum_i \epsilon_i)$-differential privacy. Second, the
  \emph{parallel composability} property states that computing each
  function on disjoint subsets provides $\mathtt{max} ( \epsilon_i
  )$-differential privacy.
\end{theorem}


\paragraph{Additively-Homomorphic Encryption.}%
\label{subsubsec:enc}%



Additively-homomorphic encryption schemes essentially allow to perform
addition operations over encrypted data. Any additively-homomorphic
encryption scheme fits our approach provided that it satisfies the
following properties. First, it must provide \emph{semantic security
  guarantees}. Stated informally, this property requires that given a
ciphertext, the public encryption key, and possible auxiliary
information about the plaintext, then no polynomial-time algorithm is
able to gain non-negligible knowledge on the
plaintext~\cite{goldreich2005foundations}. Second, it must be
\emph{additively-homomorphic}. %
Informally, given $a$ and $b$ two integers, $\mathtt{E}$ the
encryption function, $X$ the encryption key, $K$ the decryption key,
$\mathtt{D}$ the decryption function, and $+_h$ the homomorphic
addition operator, then
$\mathtt{D}_K(\mathtt{E}_X(a) +_h \mathtt{E}_X(b)) == a+b$. Third, the
scheme must support \emph{non-interactive threshold decryption}. %
We additionally use this optional property, available in some schemes
(\emph{e.g.,}\cite{paillier1999public,damgaard2001generalisation}). It
allows the decryption key to be \emph{split} in $n_K$
\emph{key-shares} ${K_i}$ such that a complete decryption requires to
perform independently $T \leq n_K$ partial decryptions by distinct
key-shares. Note that in a typical key generation setting, pairs of
keys are generated once and for all by a non-colluding, independent
entity. %

Paillier cryptosystem~\cite{paillier1999public} and its Damgard-Jurik generalization~\cite{damgaard2001generalisation} are instances of encryption schemes that provide the desired properties and are widely available. We refer the interested reader to the original papers for details.


\paragraph{Private Information Retrieval.}

In a nutshell, Private Information Retrieval (PIR) techniques allow a
client to download binary objects (\emph{e.g.,} a record, a movie)
stored on a server in a \emph{library} of objects, without revealing
to the server which of the binary objects has been downloaded. We call
this function the \texttt{PIR-get} function. Emerging PIR protocols
are now affordable and able to cope with the latency constraints of
real-life scenarios (\emph{e.g.,} in media consumption scenarios
\cite{popcorn,xpir,itpir}).
%
%
%
Our approach makes use of the security guarantees of PIR
techniques. In this paper, for concreteness, we consider a PIR protocol
based on additively-homomorphic encryption called
\emph{XPIR}~\cite{xpir}. It is part of the \emph{computational PIR}
family of protocols, that provides computational security
guarantees. %
%
However our approach could use other protocols, that provide different
tradeoffs efficiency/security (\emph{e.g.,} an
\emph{information-theoretic PIR} protocol such as~\cite{itpir} that
uses efficient bitwise XOR operators, provides information-theoretic
security, but assumes no collusion between several supporting
servers).

XPIR~\cite{xpir} considers a \emph{library} $\mathcal{L}$ of $n$
binary objects, called \emph{items} in the following. All items are
assumed to share the same length in bits, denoted $l$. The library is
stored as a matrix of $y$-bit integers: $\mathcal{L} \in
(\{0,1\}^y)^{n \times (l/y)}$. XPIR uses an additively-homomorphic
cryptosystem (see above) and implements the \texttt{PIR-get} function
as follows. XPIR assumes that each item has a unique id, and that
clients know the list of ids of the existing items in
$\mathcal{L}$. Now a client wants to retrieve the item of id $i$ and
thus calls \texttt{PIR-get(i)}. First, it instanciates a vector of $n$
bits, initializes all bits to $0$s, sets to $1$ the bit at id $i$,
encrypts each bit separately, and sends the resulting encrypted vector
- denoted $c$ - to the server. Second, for all $j$ in $[1, l/y]$, the
server computes $r_j = \prod_{i=1}^n c[i]^{L_{i,j}}$ %
(recall that the product between two encrypted integers is the way the
additively-homomorphic addition operator $+_h$ is performed by the underlying cryptosystem) and sends it back to the client. Each
$r_j$ is thus actually a sum of (1) encrypted $0$s (corresponding to
the encrypted $0$s in $c$) and (2) an encrypted bit-subsequence of the
requested binary object (corresponding to the encrypted $1$ in $c$).
Third and finally, the client decrypts each $r_j$ received - obtaining
hence the various bit-subsequences of the item requested - and
concatenates them to obtain the complete bitstring.

It is to be noticed that three main parameters may affect PIR efficiency: the size of the library itself, that has to be read for each call to the \texttt{PIR-get} function, the size of items, that impacts the size of downloads (multiplied by an \emph{expansion factor}, the ratio between the size of an encrypted value and the clear value), and the number of items (again, multiplied by the expansion factor), for the size of upload of the request.

\subsection{Space Partitionning based on KD-Trees}

A KD-Tree~\cite{kdtree} is a well-known data structure designed for
partitioning datasets in k-dimensional balanced partitions. It is
constructed by recursively dividing the space in two around the
median. It is widely used to index data. Moreover, it contains
valuable information about the data distribution (it is balanced)
without sizing individual data points. %

\subsection{Quality Measure}
We evaluate the quality of the (perturbed) output of the \texttt{PKD}
algorithm by measuring the loss of accuracy resulting from its use
rather than using non-protected raw profiles. As stated in
Definition~\ref{eq:quality}, %
given a set of tasks $T$, we compute the average absolute error
between (1) the approximate number of workers matching with each task
$t \in \mathcal{T}$ according to the (perturbed) partitions and counts
and (2) the exact non-protected number of matching workers. Note that
the error comes from the perturbation used for satisfying differential
privacy but also from the inherent approximation due to the use of a
coarse grain data structure (partitions and counts) synthesizing raw
data.

\begin{definition}[Quality]
  \label{eq:quality}
  Given a set of worker profiles $\mathcal{P}$, a set of tasks
  $\mathcal{T}$, and, for each task $t \in \mathcal{T}$, the real
  number of workers matching with it $t_{match}$ and its approximated
  value $\widetilde{t_{match}}$ according to the perturbed distribution we
  provide.  We compute the quality of the distribution as
  $$Q = \frac{1}{|\mathcal{T}|}\sum_{t\in{\mathcal{T}}}\frac{|t_{match} -
  \widetilde{t_{match}}|}{t_{match}}$$
\end{definition}

We also measure the quality of our assignment by using \emph{precision}, to size the number of tasks that are uselessly downloaded by workers, as seen in Definition~\ref{def:precision}.
Note that in our context, all matching tasks will be downloaded, such that a \emph{recall} measure is not relevant as it will always be equal to $1$.

\begin{definition}[Precision]
	\label{def:precision}
For a given assignment, we call precision the fraction of downloaded tasks that match with the worker.
This precision is computed as: $$precision=mean_{t\in \mathcal{T}}(\frac{|\{w: match(w,t) \wedge download(w,t)\}|}{|\{w: download(w,t)\}|})$$
\end{definition}

\section{The \texttt{PKD} Algorithm}
\label{sec:pkd}%

Our proposal comes from a rethinking of the centralized version of the
\emph{KD-Tree} construction algorithm~\cite{kdtree}, which is
essentially a recursive computation of medians. In our setting, each
worker holds its own, possibly sensitive, profile and no single
centralized party is trusted. Centralizing the profiles of workers in
order to compute a KD-Tree over them is therefore not possible. A naive
approach could make use of an order-preserving encryption
scheme~\cite{AKSX2004} (OPE), but %
%
%
these schemes are well-known for their low security
level~\cite{BCO2012} and especially their inherent weaknesses against
frequency analysis attacks. We rather favor sound privacy guarantees -
without sacrificing efficiency - by approaching the median through the
computation of histograms, with a computation distributed between
workers and the platform. Similarly to its centralized counterpart,
each iteration of our recursive algorithm divides the space of
skills in two around the median (of one dimension at a time) given a
perturbed histogram representing the distribution of a single skill in
the crowd. For simplicity, the current version of the algorithm
terminates after a fixed number of splits, but more elaborate
termination criteria can be defined. %

\subsection{Computing Private Medians}
\paragraph{From Private Sum to Private Histogram.}%
\label{sec:private_sum}%
We start by explaining how to perform differentially private sums
based on noise-shares and additively-homomorphic additions. %
It allows the platform to get the result of the addition of a single
bin over $n$ different workers while satisfying differential
privacy. We then show that this function is a sufficient building
block for computing perturbed histograms.

%

Let us consider a fixed $\epsilon > 0$, a maximum size of collusion
$\tau \in \mathds{N}, \tau > 0$, a set of workers $\mathcal{P}$ (of
size $|\mathcal{P}|$), and a single bin $b_{i}$ associated to a range
$\phi$.  Each worker $p_i$ holds a single private local value,
\emph{i.e.,} her skill on the dimension that is currently being split.
Initially, the bin $b_{i}$ is set to $0$ on all workers. Only the
workers whose local values fall within the range of the bin $b_{i}$
set $b_{i}$ to $1$. Our first goal is to compute the sum of
the 
bins $b_{i}$ of all workers $p_i \in \mathcal{P}$ such that no set of
participant smaller than $\tau$ can learn any information that has not
been perturbed to satisfy $\epsilon$-differential privacy.

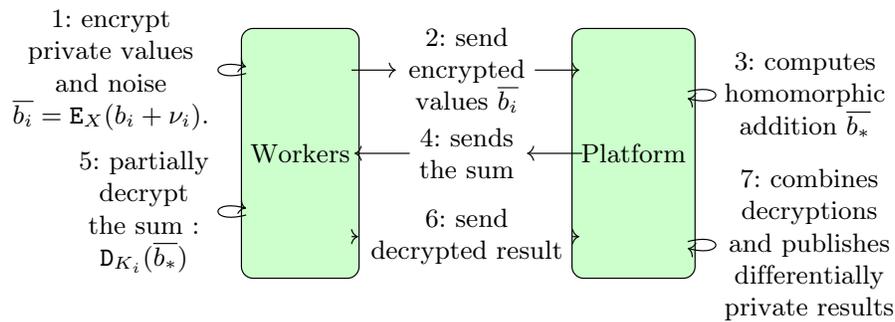
\begin{figure}
  \centering \resizebox{1\linewidth}{!}{
    \begin{tikzpicture}[every text node part/.style={align=center}]
      \tikzstyle{workers}=[rectangle,rounded corners, minimum width =
      1cm, minimum height = 3cm, text centered, draw=black,
      fill=green!20] \tikzstyle{platform}=[rectangle,rounded corners,
      minimum width = 1cm, minimum height = 3cm, text centered,
      draw=black, fill=green!20] \tikzstyle{every
        node}=[font=\footnotesize] \node (w) [workers] at (0,0)
      {Workers}; \node (pf) [platform] at (4,0) {Platform}; \node (w0)
      at (-0.5,1) {} edge [pre, loop left] node[align=center] {1:
        encrypt\\ private values \\ and noise\\ $\overline{b_{i}} =
        \mathtt{E}_X (b_{i} + \nu_{i})$.} %
      (w0) ; \node (t1) at (0.5,1){}; \node (t2) at (3.5,1){}; \node
      (m0) at (2, 1)
      {2: send\\
        encrypted\\values $\overline{b_{i}}$} edge [pre] (t1) edge
      [post] (t2); \node (pf0) at (4.5,0.7) {} edge [pre, loop right]
      node[align=center] {3: computes\\homomorphic\\addition
        $\overline{b_*}$} (pf0) ; \node (t3) at (0.5,0){}; \node (t4)
      at (3.5,0){}; \node (m1) at (2, 0) {4: sends\\the sum} edge
      [post] (t3) edge [pre] (t4); \node (w1) at (-0.5,-0.7) {} edge
      [pre, loop left] node[align=center] {5: partially\\ decrypt\\the
		  sum :\\ $\mathtt{D}_{K_i} (\overline{b_*})$} (w1) ; \node (t5) at
      (0.5,-1){}; \node (t6) at (3.5,-1){}; \node (m2) at (2, -1) {6:
        send\\decrypted result} edge [pre] (t5) edge [post] (t6);
      \node (pf1) at (4.5,-1.1) {} edge [pre, loop right]
      node[align=center] {7: combines\\ decryptions\\ and
        publishes\\differentially\\private results} (pf1) ;

  \end{tikzpicture}
}
\caption{Computing a Private Sum} %
\label{fig:sum}
\end{figure}

The privacy-preserving sum algorithm, depicted in Fig.~\ref{fig:sum},
considers that keys have been generated and distributed to workers,
such that $T > \tau$ key-shares are required for decryption (see
Section~\ref{subsubsec:enc}). The algorithm consists in the following
steps:
\begin{description}
\item [Step 1 (each worker) - Perturbation and Encryption] %
  %
  First, each worker perturbs its value by adding a noise-share,
  denoted $\nu_{i}$, to it. Noise-shares are randomly generated
  locally such that the sum of $|\mathcal{P}|-\tau$ shares satisfies
  the two-sided geometric distribution (see Definition~\ref{def:geom}
  for the geometric distribution, and Theorem~\ref{the:inf2geom} for
  its infinite divisibility). %
  Note that noise-shares are overestimated\footnote{We require that
    the sum of $|\mathcal{P}|-\tau$ noise-shares be enough to satisfy
    differential privacy but we effectively sum $|\mathcal{P}|$
    noise-shares. Note that summing more noise-shares than necessary
    does not jeopardize privacy guarantees.} to guarantee that the
  final result is differentially private even for a group of up to
  $\tau$ workers sharing their partial knowledge of the total noise
  (their local noise-share).  Each worker then encrypts
  $b_{i} + \nu_i$ by the additively-homomorphic encryption scheme in
  order to obtain $\overline{b_{i}}$ :
  $\overline{b_{i}} = \mathtt{E}_X (b_{i} + \nu_{i})$. %
\item [Step 2 (platform) - Encrypted Sum] %
  The platform sums up together the encrypted values received~:
  $\overline{b_*} = \sum_{\forall i} \overline{b_{i}}$ where the sum
  is based on the additively-homomorphic addition $+_h$. %
\item [Step 3 (subset of workers) - Decryption] %
  The platform finally sends the encrypted sum $\overline{b_*}$ to at
  least $T$ distinct workers. Upon reception, each worker partially
  decrypts it based on her own key-share - $\mathtt{D}_{K_i} (\overline{b_*})$ -
  and sends it back to the platform. The platform combines the partial
  decryptions together and obtains $\widetilde{b_*}$, \emph{i.e.,} the
  differentially private sum of all private bins $b_{i}$.
\end{description}

\begin{algorithm}[tbhp]
  \KwData{\\
    $\mathcal{P}$: Set of workers \\
    $\mathcal{D}_{min}, \mathcal{D}_{max}$: Definition domain of the private local value of workers \\
    $l$: Number of bins \\
    $(\phi_0, \ldots, \phi_{l-1})$ : the $l$ ranges of the bins \\
    $X$ : public encryption key (same for all workers). %
    \\
    $\{K_i\}$ : private decryption keys (one per worker). \\
    $\epsilon_{m}$: differential privacy budget for this iteration \\
    $\tau$: maximum size of a coalition ($\tau < |\mathcal{P}|$) \\
    $T$: number of key-shares required for decryption ($T > \tau$)
  }%
  \KwResult { %
    $\widetilde{m}$: estimate of the median of the workers' local private
    values %
  }%

  \For{%
    all workers $p_i \in \mathcal{P}$%
  }{ %

    Compute $l$ noise shares $\nu_{i,j} = R_1 - R_2$, where
    $0\leq j < l$ and $R_1$ and $R_2$ are independent identically
    distributed random variables with probability density function
    $g(k)=\binom{k-1+\frac{1}{|\mathcal{P}|-\tau}}{k}(e^{-\epsilon_m})^k(1-e^{-\epsilon_m})^{\frac{1}{|\mathcal{P}|-\tau}}$
    \\
    Set the value of the bin
    $\overline{b_{i,k}}= \mathtt{E}_X (1+\nu_{i,k})$, where $\phi_k$
    is the histogram range within which the local value
    of the worker falls. \\
    Set the value of the other bins to:
    $\overline{b_{i,j}} = \mathtt{E}_X (\nu_{i,j})$, $j\neq k$.
    \\

  }%
  \textbf{Platform :} sum the encrypted bins at the same index
  received from different workers in order to obtain the encrypted
  perturbed histogram :
  $(\overline{b_{*,0}} = \sum_{\forall i} \overline{b_{i, 0}}, \ldots,
  \overline{b_{*,l-1}} = \sum_{\forall i} \overline{b_{i,
      l-1}})$. \\
  \textbf{Workers (T distinct workers) :} Decrypt partially the
  encrypted perturbed histogram bin per bin and send the resulting
  partially decrypted histogram to the platform :
  $(\mathtt{D}_{K_i}(\overline{b_{*,0}}), \ldots,
  \mathtt{D}_{K_i}(\overline{b_{*,l-1}}))$. \\
  \textbf{Platform :} Combine the partial decryptions together to
  obtain the decryption of the histogram and estimate the median
  $\widetilde{m}$ according to Equation~\ref{eq:median}.



  \Return{$\widetilde{m}$}

  \caption{\texttt{PrivMed} : Privacy-Preserving Estimation of the
    Median in the \texttt{PKD} algorithm}%
  \label{algo:priv_median}
\end{algorithm}

Now, assuming that the histogram format is fixed beforehand -
\emph{i.e.,} number $l$ of bins and ranges $(\phi_0, \ldots,
\phi_{l-1})$ - it is straightforward to apply the private sum
algorithm on each bin for obtaining the perturbed histogram based on
which the median can then be computed. %
For example, in order to get a histogram representing the
distribution of skill values for, \emph{e.g.,} \texttt{Python
  programming}, and assuming a basic histogram format - \emph{e.g.,}
skill values normalized in $[0, 1]$, $l=10$ bins, ranges
$(\phi_0=[0,0.1[, \ldots, \phi_9=[0.9,1])$ - it is sufficient to
launch ten private sums to obtain the resulting perturbed $10$-bins
histogram.

\paragraph{\texttt{PrivMed}: privacy-preserving median computation.} %
The histogram computed based on the privacy-preserving sum algorithm 
can be used by the platform to estimate the value of the median around
which the split will be performed. When by chance the median falls
precisely between two bins (\emph{i.e.,} the sum of the bins on the
left is exactly $50\%$ of the total sum, same for the bins on the
right) its value is exact. But when the median falls within the range
of one bin (\emph{i.e.,} in any other case), an additional hypothesis
on the underlying data distribution within the bin must be done in
order to be able to estimate the median.  For simplicity, we will
assume below that the distribution inside each bin is uniform but a
more appropriate distribution can be used if known. %

\begin{equation}\label{eq:median}
  \widetilde{m} = \mathcal{D}_{min} + \frac{\mathcal{D}_{max}-\mathcal{D}_{min}}{l}\cdot(k + \frac{1}{2} + \frac{\theta_> - \theta_<}{2\cdot \widetilde{b_{*,k}}})
\end{equation}%
The resulting \texttt{PrivMed} algorithm is detailed in Algorithm~\ref{algo:priv_median}.

Let's consider the histogram obtained by the private sum algorithm. It
is made of $l$ bins denoted $(\widetilde{b_{*,0}}, \ldots, \widetilde{b_{*, l-1}})$, and each
bin $\widetilde{b_{*,j}}$ is associated to a range $\phi_j$. The ranges partition
a totally ordered domain ranging from $\mathcal{D}_{min}$ to
$\mathcal{D}_{max}$ (\emph{e.g.,} from $\mathcal{D}_{min} = 0$ to
$\mathcal{D}_{max}=1$ on a normalized dimension that has not been
split yet). Let $\phi_{k}$ denote the 
range containing the median, 
$\theta$ denote the sum of all the bins - \emph{i.e.,} $\theta =
\sum_{i<l} \widetilde{b_{*,i}}$ - %
and $\theta_{<}$ (resp. $\theta_>$) the sum of the bins that are
strictly before (resp. after) $\widetilde{b_{*,k}}$ - \emph{i.e.,} $\theta_< =
\sum_{i<k} \widetilde{b_{*,i}}$ (resp. $\theta_> = \sum_{i>k} \widetilde{b_{*,i}}$). %
Then, an estimation $\widetilde{m}$ of the median can be computed as
follows\footnote{Note that in the specific case where the median falls within a bin
equal to $0$ (\emph{i.e.,} $\widetilde{b_{*,k}}=0$), then any value
within $\phi_k$ is equivalent.}: %

\subsection{Global Execution Sequence}
\label{sec:global_exec}

Finally, the Privacy-preserving KD-Tree algorithm, \texttt{PKD} for
short, performs the median estimation described above iteratively
until it stops and outputs (1) a partitioning of the space of skills
together with (2) the perturbed number of workers within each
partition. The perturbed number of workers is computed by using an
additional instantiation of the private sum algorithm when computing
the private medians\footnote{Note that the perturbed histograms could
  have been used for computing these counts but using a dedicated
  count has been shown to result in an increased precision.}. We focus
below on the setting up of the parameters of the various iterations,
and on the possible use of the resulting partitions and counts by the
requesters. An overview is given in Algorithm~\ref{algo:priv_kdtree}.

\paragraph{Main Input Parameters.}%
\label{sec:distrib_epsilon}%



Intuitively, the privacy budget $\epsilon$ input of the \texttt{PKD} algorithm
sets an upper bound on the information disclosed along the complete
execution of the algorithm - it must be \emph{shared} among the
various data structures disclosed.
Thus, each iteration inputs a portion of the initial privacy budget
such that the sum of all portions remains lower than or equal to
$\epsilon$ - see the composability property in
Theorem~\ref{the:composition}. Computing a good budget allocation in a
tree of histograms is a well-known problem tackled by several related
works~\cite{psd,qardaji2013differentially}. %
In this work, we simply rely on existing privacy budget distribution
methods. For example, based on \cite{psd}, $\epsilon$ is divided as
follows. First, $\epsilon$ is divided in two parts: one part, denoted
$\epsilon^m$, is dedicated to the perturbations of the medians
computations (\emph{i.e.,} the bins of the histograms), and the other
part, denoted $\epsilon^c$, is dedicated to the perturbation of the
number of workers inside each partition.  Second, a portion of each of
these parts is allocated to each iteration $i$ as follows. For each
iteration $i$ such that $0 \leq i \leq h$, where $h$ is the total
number of iterations (\emph{i.e.,} the height of the tree in the
KD-Tree analogy), the first iteration is $h$ (\emph{i.e.,} the root of
the tree) and the last one is $0$ (\emph{i.e.,} the leaves of the
tree) :
\begin{equation}
  \epsilon^c_i =
  2^{(h-i)/3}\epsilon^c\frac{\sqrt[3]{2}-1}{2^{(h+1)/3}-1}
\end{equation}

\begin{equation}
  \epsilon^m_i = \frac{\epsilon^m}{h}
\end{equation}

Note that similarly to~\cite{psd}, we set the distribution of
$\epsilon$ between $\epsilon^c$ and $\epsilon^m$ as follows:
$\epsilon^c = 0.7\cdot \epsilon$ and $\epsilon^m = 0.3\cdot
\epsilon$. Other distributions could be used.

%


The \texttt{PKD} algorithm stops after a fixed number of iterations
known beforehand. Note that more elaborate termination criteria can be
defined (\emph{e.g.,} a threshold on the volume of the subspace or on
the count of worker profiles contained). The termination criteria
must be chosen carefully because they limit the number of splits of
the space of skills and consequently the number of dimensions of
worker profiles that appear in the final subspaces. Ideally, the
termination criteria should allow at least one split of all
dimensions. However, this may not be possible or suitable in practice
because of the limited privacy budget. In this case, similarly to a
composite index, a sequence of priority dimensions must be chosen so
that the algorithm splits them following the order of the
sequence. 
The dimensions that are not part of the sequence will simply be
ignored. Note that the number of dimensions in worker profiles, and
their respective priorities, is closely related to the application
domain (\emph{e.g.,} %
How specific does the crowdsourcing process need to be ?). %
%
In this paper, we make no assumption on the relative importance of
dimensions.

\begin{algorithm}[tbhp]
  \KwData{\\
    $\mathcal{P}$: Set of workers \\
	$E$ the current space of skills of $d$ dimensions
    \\
	$h$: height of the KD-Tree
  }%
  \KwResult { %
    $T$: A Privacy-preserving KD-tree with approximate counts of workers for each leaf
  }%

  We create $T$ as a single leaf, containing the whole space $E$ and a count of all workers.

  \While{current height is smaller than final height}
  {
	  Choose a dimension $d$ (for exemple, next dimension).\\
	  \For{all leaves of the current tree $T$}{

	  Compute $m$ the private median of the space of the leaf, as explained in Section~\ref{sec:private_sum}.\\
	  For both subspaces separated by the median, compute a private count as explained in Section~\ref{sec:private_sum}.\\
  	  Create two leaves, containing the two subspaces and associated counts. \\
	  Replace the current leaf by a node, containing the current space and count, and linking to the two newly created leaves.




	}
	Increment the current height.


  }

  Apply post-processing techniques explained in Section~\ref{sec:global_exec}.

  \Return{Tree $T$}

  \caption{The \texttt{PKD} Algorithm}%
  \label{algo:priv_kdtree}
\end{algorithm}

\paragraph{Post-Processing the Output.}
Considering the successive splits of partitions, we can enhance the
quality of the counts of workers by exploiting the natural constraints
among the resulting tree of partitions : we know that the number of
workers in a parent partition must be equal to the number of workers
in the union of its children. %
\emph{Constrained inference techniques} have already been studied as a
post-processing step to improve the quality of trees of perturbed
histograms, first in \cite{hay2010boosting} and then improved in
\cite{psd} which adapts the method to non-uniform distribution of
budget. These constrained inference techniques can be used in our
context in a straightforward manner in order to improve on the quality
of the resulting partitioning. We refer the interested reader to
\cite{hay2010boosting,psd} for details.

\subsection{Complexity Analysis}%
\label{sec:compl-analysis}





We evaluate here the complexity of the \texttt{PKD} algorithm with
respect to the number of encrypted messages computed and sent both to and by the platform. The results are summed up in
Table~\ref{tab:cost}. %

The first step to consider is the number of partitions created in the
KD-tree. Seen as an index (with one leaf for each point), the
construction of a KD-tree requires $2^{h+1}-1$ nodes, including
$2^h-1$ internal nodes, where $h$ is the maximum height of the
KD-Tree. For each node, an encrypted sum is performed, and for each
internal node, a histogram is additionally computed, which require
$l$ sums, for a total of $(2^{h+1}-1) + (l\cdot(2^h-1))$ sums.  These
counts all require the participation of every worker: for each count,
$|\mathcal{P}|$ encrypted messages are computed and sent.

The platform also sends back encrypted messages for each sum, for
decryptions to be performed. For each sum, it sends at least $T$ times
the homomorphically computed sum, where $T$ is the threshold number of
key-shares required for decryption. For simplicity, we assume that the
platform sends the cyphertexts to $T$ workers (these are the only
encrypted messages that have to be sent to workers during this
protocol). Each contacted worker then answers by an encrypted value
(the partial decryption). As a conclusion, the total number of
encrypted values sent by the workers to the platform
$\mathcal{M}_{\Sigma w}$ is:
\begin{equation}
  \label{eq:w-send-total}%
  \mathcal{M}_{\Sigma w}=(|\mathcal{P}| + T) \cdot ( l \cdot (2^h-1) + (2^{h+1}-1))
\end{equation}

However, as our computation is distributed among all workers, each
worker only sends fewer encrypted messages on average
$\overline{\mathcal{M}_w}$.
\begin{equation}
    \label{eq:w-send-avg}%
\overline{\mathcal{M}_w}=(1+\frac{T}{|\mathcal{P}|}) \cdot (l \cdot (2^h-1) + (2^{h+1}-1))
\end{equation}


Finally, the platform sends $\mathcal{M}_{pf}$ encrypted messages.
\begin{equation}
  \label{eq:pf-send}%
  \mathcal{M}_{pf}=T \cdot (l \cdot (2^h-1) + (2^{h+1}-1))
\end{equation}

\begin{table}
  \centering
  \begin{tabular}{|l|l|}
    \hline
    To the platform & $(|\mathcal{P}| + T) \cdot ( l \cdot (2^h-1) + (2^{h+1}-1))$ \\
    \hline
    By worker (avg) & $(1+\frac{T}{|\mathcal{P}|})\cdot(l\cdot(2^h-1) + (2^{h+1}-1))$ \\
    \hline
    By the platform & $T\cdot(l\cdot(2^h-1) + (2^{h+1}-1))$\\
    \hline
  \end{tabular}
  \caption{Number of encrypted messages sent. $\mathcal{P}$ is the set
    of workers, $T$ the number of partial keys required for
    decryption, $h$ the depth of the KD-tree, and $l$ the number of
    bins per median}
  \label{tab:cost}
\end{table}

\subsection{Security Analysis}
The only part of the \texttt{PKD} algorithm that depends on raw data
is the private sum. The security analysis thus focuses on proving that
a single private sum is secure, and then uses the composability
properties (see
Theorem~\ref{the:composition}). Theorem~\ref{thm:sum-secure} proves
that the privacy-preserving sum algorithm is secure. We use this
intermediate result in Theorem~\ref{thm:kd-secure} to prove the
security of the complete \texttt{PKD} algorithm. %

\begin{theorem}[Security of the privacy-preserving sum algorithm]
  \label{thm:sum-secure}
  The privacy-preserving sum algorithm %
  satisfies \texttt{$\epsilon_\kappa$-SIM\--\-CDP} privacy against
  coalitions of up to $\tau$ participants.
\end{theorem}

\begin{proof}\emph{(sketch)} %
  First, any skill in a profile of a participating worker is first
  summed up locally with a noise-share, and then encrypted before
  being sent to the platform. We require the encryption scheme to
  satisfy semantic security, which means that no
  computationally-bounded adversary can gain significant knowledge
  about the data that is encrypted. In other words, the leak due to
  communicating an encrypted data is negligible. Second, the
  homomorphically-encrypted additions performed by the platform do not
  disclose any additional information. Third, the result of the
  encrypted addition is decrypted by combining $T>\tau$ partial
  decryptions, where each partial decryption is performed by a
  distinct worker. The threshold decryption property of the encryption
  scheme guarantees that no coalition of participants smaller than
  $\tau$ can decrypt an encrypted value, and the honest-but-curious
  behaviour of participants guarantees that no other result but the final one will be decrypted (\emph{e.g.} the platform does not ask for a decryption of a value that would not have been sufficiently perturbed).
  The final sum consists in the sum of all private values, to which
  are added $|\mathcal{P}|$ noise-shares. These shares are computed
  such that the addition of $|\mathcal{P}|-\tau$ shares is enough to
  satisfy $\epsilon$-differential privacy. Thanks to the
  post-processing property of differential privacy, adding noise to a
  value generated by a differentially-private function does not impact
  the privacy level. The addition of $\tau$ additional noise-shares
  consequently allow to resist against coalitions of at most $\tau$
  participants without thwarting privacy. As a result, since the
  privacy-preserving sum algorithm is the composition of a
  semantically secure encryption scheme with an
  $\epsilon$-differentially private function, it is computationally
  indistinguishable from a pure differentially private function, and
  consequently satisfies \texttt{$\epsilon_\kappa$-SIM-CDP} privacy
  against coalitions of up to $\tau$ participants.
    %
\end{proof}

\begin{theorem}[Security of the \texttt{PKD} algorithm] %
  \label{thm:kd-secure}
  The \texttt{PKD} algorithm satisfies
  \texttt{$\epsilon_\kappa$-SIM-CDP} privacy against coalitions of up
  to $\tau$ participants. %
\end{theorem}

\begin{proof} \emph{(sketch)} %
  In the \texttt{PKD} algorithm, any collected information is
  collected through the \texttt{PrivMed} algorithm based on the
  privacy-preserving sum algorithm. Since (1) the privacy-preserving
  sum algorithm satisfies \texttt{$\epsilon_\kappa$-SIM-CDP} (see
  Theorem~\ref{thm:sum-secure}) against coalitions of up to $\tau$
  participants, (2) \texttt{$\epsilon_\kappa$-SIM-CDP} is composable
  (see Theorem~\ref{the:composition}), and (3) the privacy budget
  distribution is such that the total consumption does not exceed
  $\epsilon$ (see Section~\ref{sec:distrib_epsilon}), it follows
  directly that the \texttt{PKD} algorithm satisfies
  \texttt{$\epsilon_\kappa$-SIM-CDP} against coalitions of up to
  $\tau$ participants.
\end{proof}


%

\section{Privacy-Preserving Task Assignement}
\label{sec:delivery}
Once the design of a task is over, it must be assigned to relevant
workers and delivered. Performing that while satisfying differential
privacy and at the same time minimizing the number of downloads of the
task's content is surprisingly challenging. We already discarded in
Section~\ref{sec:intro}, for efficiency reasons, the
\emph{spamming} approach in which each task is delivered to all
workers.
More elaborate approaches could try to let the platform
filter out irrelevant workers based on the partitioned space output
by the \texttt{PKD} algorithm (see the Section~\ref{sec:pkd}). The partitioned
space would be used as an index over workers in addition to its
primary task design usage. For example, workers could subscribe to
their areas of interest (\emph{e.g.,} by sending an email address to
the platform together with the area of interest) and each task would
be delivered to a small subset of workers only according to its
metadata and to the workers' subscriptions. However, despite their
appealing simplicity, these \emph{platform-filtering} approaches
disclose unperturbed information about the number of workers per area,
which breaks differential privacy, and fixing the leak seems hard
(\emph{e.g.,} random additions/deletions of subscriptions, by
distributed workers, such that differential privacy is satisfied and
the overhead remains low).

We propose an alternative approach, based on Private Information
Retrieval (PIR) techniques, to diminish the cost of download on the
workers side, while preserving our privacy guarantees.

\subsection{PIR for Crowdsourcing: challenges and naive approaches}\label{sec:PIR}\label{subsec:challengesPIR}%

The main challenge in applying PIR in our context consists in designing a PIR-library such that no information is disclosed during the retrieval of information, and performance is affordable in real-life scenarios.
To help apprehending these two issues, we here present two naive methods that break these conditions and show two extreme uses of PIR: one efficient but unsecure, the other is secure but unefficient.

A first PIR-based approach could consist in performing
straightforwardly a PIR protocol between the workers and the platform, while considering the PIR-library as the set of tasks itself.
The platform maintains a key-value map that stores the complete set of tasks (the values, bitstrings required to perform the tasks) together with a unique identifier per task (the keys), together with their metadata.
The workers download the complete list of tasks identifiers and metadata, select locally the identifiers
associated to the metadata that match with their profiles, and launch one
\texttt{PIR-get} function on each of the selected
identifiers.
However, this naive approach leads to blatant security
issues through the number of calls to the \texttt{PIR-get}
function.
Indeed, in some cases, the platform could deduce the precise number of workers within a specific subspace of the space of skills: with the knowledge of the number of downloads for each worker\footnote{Even if the identity of workers is not directly revealed, it is possible to match downloads together to break \emph{unlinkability} and deduce these downloads come from the same individual, for example by using the time of downloads, cookies or other identification techniques}, it is possible to deduce, for each $k$, the number of workers downloading $k$ tasks.
From that, the platform can deduce that the number of workers located in subspaces where $k$ tasks intersect together, and therefore precise information on their skills. This information, kept undisclosed thanks to the \texttt{PKD} algorithm, breaks differential privacy guarantees.

A secure but still naive approach could be to consider the power set of the set of tasks as the PIR-library, with padding to all file such that they are all the same size (in bits).
After this, a worker chooses the PIR-object corresponding to the set of tasks she intersects with, and uses \texttt{PIR-get} on it.
Although this method prevents the previously observed breach to appear (all behaviours are identical to the platform since everyone downloads exactly one PIR-object, and all PIR-objects are of the same size), this method would lead to extremely poor results: as every object of the library is padded to the biggest one, and the biggest set of the super set of tasks is the set of tasks itself, this algorithm is even worse than the spamming approach (everyone downloads at least as much as the sum of all tasks, with computation overheads).

These two naive uses of PIR illustrate two extreme cases: the first one shows that using PIR is not sufficient to ensure privacy, and the second one illustrates that a naive secure use can lead to higher computation costs than the spamming approach.
In the following, we introduce a method to regroup tasks together, such that each worker downloads the same number of PIR items (to achieve security), while mitigating performance issues by making these groups of tasks as small as possible.

\subsection{PIR partitioned packing}


The security issue showed in the naive PIR use comes from the fact that the number of downloads directly depends on the profiles of workers.
Indeed, as the platform has access to the number of downloads, this link leaks information about workers' skills.
In order to break this link, we propose to ensure that each worker downloads the same number of items, whatever their profile is.
For simplicity, we fix this number to $1$\footnote{In general, more files can be downloaded at each worker session, but this does not impact
  significantly the overall amount of computation and does not impact
  at all the minimum download size for workers.}, and call \emph{packing} a PIR library that allows each worker to retrieve all their tasks with only one item, and \emph{bucket} an item of such a library, as seen in Definition~\ref{def:packing}. We prove in Theorem~\ref{proof:sec_packing} that any packing fulfills our security model.

We can now formalize the conditions that a PIR library must fulfill in order to both satisfy privacy and allow any worker to download all the tasks she matches with.

\begin{definition}[Packing, Bucket]\label{def:packing}

A packing $L$ is a PIR library which fulfills the following conditions:
\begin{enumerate}
\item \label{enum:number1}
	\textbf{Security condition} Each worker downloads the same number of buckets. This number is set to $1$.
\item \label{enum:packcondsize}
  \textbf{PIR requirement} Each PIR item has the same size in bits (padding is allowed): $$\forall b_1, b_2 \in L, ||b_1|| = ||b_2||$$
  This condition comes from the use of PIR.
\item \label{enum:packcondutil}
  \textbf{Availability condition} For all points in the space of skills, there has to be at least one item containing all tasks matching with this point. In other words, no matter their skills (position in the space), each worker can find a bucket that provides every task they match with.
\end{enumerate}
A \textbf{bucket} $b \in L$ is an item of a packing.
We note $|b|$ for the number of tasks contained in the bucket $b$, and $t \in b$ the fact that a task $t$ is included in bucket $b$.
The size in bits of a bucket $b$ is denoted as $||b||$.
\end{definition}

\begin{theorem}\label{proof:sec_packing}
  The use of PIR with libraries which fulfill the packing conditions satisfy \texttt{$\epsilon_\kappa$-SIM\--\-CDP} privacy against coalitions of up to $\tau$ participants.
\end{theorem}

\begin{proof}\emph{(sketch)}
In order to prove the security of packing, we observe
that (1) the XPIR protocol has been proven computationally secure in~\cite{xpir}, such that it satisfies $\epsilon_\kappa$\texttt{-SIM-CDP}, and
(2) the use of packing prevents any sensitive information on workers to leak through the number of downloads.
Indeed, Condition~\ref{def:packing}.\ref{enum:number1} (security) makes each worker call the \texttt{PIR-get} function only once, such that the behaviours of any two workers are indistinguishable.
Therefore, the number of \texttt{PIR-get} calls does not depend on profiles.
More precisely, the number of \texttt{PIR-get} can only leak information on the number of workers (which is bigger than or equal to the number of \texttt{PIR-get} calls), which does not depend on their profiles, and is already known by the platform.

\end{proof}

Before considering how to design an efficient packing scheme, we highlight a few noticeable implications of these conditions.
First, due to Condition~\ref{def:packing}.\ref{enum:packcondutil} (availability), any worker is matched with at least one bucket.
To simplify this model, we propose to focus on a specific kind of packings, that can be seen as a partitioning of the space, where each bucket can be linked to a specific subspace, and where all points are included in at least one of such a subspace.
We call \emph{partitioned packing} such a packing (Definition~\ref{def:partitioned}).

\begin{definition}[Partitioned Packing]\label{def:partitioned}
A partitioned packing is a packing that fulfills the following conditions:
\begin{enumerate}
	\item \label{cond:part1} Each bucket is associated with a subspace of the space of skills.
	\item \label{cond:part2} A bucket contains exactly the tasks that intersect with the subspace it is associated with (this means that all workers in this subspace will find at least the task they match with in the bucket)
	\item
	\label{cond:part3}
	Subspaces associated with the buckets cover the whole space (from Condition~\ref{def:packing}.\ref{enum:packcondutil} (availability)).
	\item
	\label{cond:part4}
		Subspaces associated with the buckets do not intersect each other
\end{enumerate}
\end{definition}


In the following, we will focus on partitioned packing.
However, in order not to lose generality, we first prove that these packings do not impact efficiency.
Indeed, when it comes to the design of a PIR library, efficiency can be affected by two main issues: the number of items and the size of the largest item (in our case, bucket) impact the communication costs, while the size of the overall library impacts the computation time on the platform.
Note that the size of the overall library is equal to the product of the size of items by their number.
We show in Theorem~\ref{thm:packing_partition} that with any packing, we can build a partitioned packing that is equivalent or better.

To prove this theorem, we introduce a specific kind of packing that we call \emph{consistent packing}, defined in Definition~\ref{def:consist}. Essentially, a consistent packing is a packing where no useless task is added to any bucket: in all buckets $b$, all tasks match with at least one point (a possible worker profile) which has all her tasks in the bucket $b$. As a result, a consistent packing avoids cases where tasks are in a bucket, but no worker would download it as the bucket does not match with all their needs.

\begin{definition}[Consistent packing]\label{def:consist}
A packing $P$ is called consistent if and only if, for all buckets $b \in P$, for all tasks $t\in b$, there exists at least one point $w$ in the subspace of $t$ such that all tasks matching with $w$ are in $b$:
$$\forall b \in P, \forall t \in b, \exists w, (match(w, t) \wedge \forall t' \in T, match(w, t') \Rightarrow t' \in b ) $$
\end{definition}

\begin{theorem}\label{thm:packing_partition}
  For any packing $P$ of tasks, there exists a partitioned packing that either has the same size of buckets, number of buckets, or smaller ones.
\end{theorem}

\begin{proof}\emph{(sketch)}
  Let $P$ be a packing of tasks that is not partitioned.
  To prove that a partitioned packing can be created that is more efficient than $P$, we distinguish two cases.
  First, we consider that each bucket of $P$ can cover a subspace, containing exactly the tasks that intersect with that subspace (thus fulfilling Conditions~\ref{def:partitioned}.\ref{cond:part1} and~\ref{def:partitioned}.\ref{cond:part2}).
  Then, we prove that any consistent packing (as in Definition~\ref{def:consist}) fulfills Condition~\ref{def:partitioned}.\ref{cond:part2}.
  After that, we consider the case where $P$ does not fulfill this condition, and create a new, smaller packing $P_f$ from $P$ that is consistent, and therefore fulfills Condition~\ref{def:partitioned}.\ref{cond:part2}, and use previous results.

  We first consider the case where all buckets of $P$ can cover a subspace while fulfilling Condition~\ref{def:partitioned}.\ref{cond:part2}, meaning that each bucket contains exactly the tasks that intersect with the subspace it covers.
  In that case, Condition~\ref{def:partitioned}.\ref{cond:part1} is trivially fulfilled.
If Condition~\ref{def:partitioned}.\ref{cond:part3} is not fulfilled, this means that there is at least a subspace that is not covered by the packing $P$.
Let $w$ be a point in such a subspace.
Since $P$ is a packing, the Condition~\ref{def:packing}.\ref{enum:packcondutil} (availability) makes it possible to match with any point of the space with at least one bucket. In particular, $w$ can be matched with a bucket $b$. It is enough to extend the subspace associated with $b$ such that it includes $w$ (note that this extension does not break Condition~\ref{def:partitioned}.\ref{cond:part2}.
We can proceed that way for any point (or more likely any subspace) that is not covered by a subspace, to associate subspaces to a bucket of $P$, such that this matching fulfills Conditions~\ref{def:partitioned}.\ref{cond:part1}, \ref{def:partitioned}.\ref{cond:part2} and \ref{def:partitioned}.\ref{cond:part3}.
If, in this matching, two subspaces associated with buckets of $P$ intersect, it is trivial to reduce one of them to fulfill Condition~\ref{def:partitioned}.\ref{cond:part4} too.
Therefore, if all buckets of $P$ can be matched with a subspace while fulfilling Condition~\ref{def:partitioned}.\ref{cond:part2}, the theorem holds, since $P$ is equivalent to a partitioned packing.


It can be noticed that if a packing is consistent (Definition~\ref{def:consist}), Condition~\ref{def:partitioned}.\ref{cond:part2} is fulfilled.
Indeed, if $\forall b \in P, \forall t \in b, \exists w, (match(w, t) \wedge \forall t' \in T, match(w, t') \Rightarrow t' \in b ) $ (Definition~\ref{def:consist}), then, for all $b$ in $P$, we can take $V$ as the union of the $|b|$ points $w$ described in the equation, one for each task $t$ in $b$.
In that case, Condition~\ref{def:partitioned}.\ref{cond:part2} is fulfilled: for each bucket $b$ and its associated subspace $V$, all tasks in $b$ intersect with $V$ (by definition, as we took $V$ as the union of one point in each task in $b$), and $b$ contains all tasks that intersect with $V$ (again, by definition, as each task $t'$ that match with a point of $V$ are in $b$).

In other words, and using the above case, making a packing consistent is sufficient to create a partitioned packing.

We now consider the case where at least one bucket $b$ of $P$ does not cover a subspace such that Condition~\ref{def:partitioned}.\ref{cond:part2} does not stand. In particular, $P$ is not consistent.
This means that there is at least one task $t \in b$ such that for all points $w$ in the subspace of $t$, there is at least one task $t'$ with which $w$ matches and that is not contained by the bucket $b$.
In other words, no points in $t$ can be matched with the bucket $b$, as $b$ lacks at least one task for each point of $t$.
As a consequence of the Condition~\ref{def:packing}.\ref{enum:packcondutil} (availability), this means that all points in $t$ are matched with another bucket.
  Therefore, the task $t$ can be removed from bucket $b$, without breaking the properties of a packing, and without increasing the number of buckets, the minimal size of buckets.
  We proceed so, by removing all such tasks in all buckets recursively:
  this trivially ends thanks to the finite number of tasks and
  buckets.
  By construction, the final packing $P_f$ is consistent.

  Therefore, packing $P_f$ is smaller than $P$, and fulfills Condition~\ref{def:partitioned}.\ref{cond:part2}, and we proved in the first case that a packing that fulfills this condition is equivalent to a partitioned packing, so $P_f$ is equivalent to a partitioned packing.
\end{proof}

\subsection{Optimizing the packing}

With this secure partitioned packing approach, we can discuss how to optimize the overall complexity.
First, it can be noticed that Conditions~\ref{def:packing}.\ref{enum:packcondsize} and~\ref{def:packing}.\ref{enum:packcondutil} (PIR requirement and availability) set a minimal size of bucket: according to Condition~\ref{def:packing}.\ref{enum:packcondutil} (availability), there has to be a bucket containing
the largest (in bits) intersection of tasks, and Condition~\ref{def:packing}.\ref{enum:packcondsize} (PIR requirement) prevents any bucket from being smaller.
Furthermore, this minimum is reachable if we consider a packing that creates a
partition for each different intersection of tasks and pad to the
largest one.
However, by building a different bucket for all the possible intersections of tasks, this packing strategy is likely to lead to a very large number of buckets (\emph{e.g.} if a task's subspace is included in another, this packing leads to two buckets instead of one: one containing both tasks, and the other containing only the largest one as it is a different intersection), while we would like to minimize it (and not only the size of buckets).
Therefore, although this packing scheme reaches the minimal size of buckets, we cannot consider it as optimal.
However, it illustrates what we call an \emph{acceptable} packing (Definition~\ref{def:acceptable}), which will be used to define optimality: a packing in which the size of buckets is minimal.

\begin{definition}[Acceptable partitioned packing]\label{def:acceptable}

	Let $E$ be a multi-dimensional space, $T$ a set of tasks, \emph{i.e.} a set of positively weighted hyper-rectangles (the hyper-rectangle is the volume of the task, and the weight is their size in bits, denoted $w_t$ for $t \in T$) in the space $E$ and $P$ a packing of these tasks.
	We call weight of a packing $w_P$ the size in bits of a bucket in $P$ (due to Condition~\ref{def:packing}.\ref{enum:packcondsize} (PIR requirement), this size is unique).
We call weight of a point $w_p$ in $E$ the sum of the weights of all tasks in $T$ which match with $p$.

We call minimal weight $m_T$ of the set of tasks $T$ the maximum weight of a point in $E$: it is the maximum size a worker could require to download.
A partitioning $P$ is called acceptable for $T$ if the size of $P$ is equal to $m_T$: $m_T = w_P$.
\end{definition}

\subsubsection{NP-hardness of optimal packing}
To define optimality, we take this minimum size of buckets, but also try to minimize the number of buckets (or in an equivalent way, the size of the PIR library), as expressed in Definition~\ref{def:optimal}.

\begin{definition}[Optimal partitioned packing]\label{def:optimal}
	For a set of tasks $T$, we call optimal packing an acceptable packing that minimizes the number of buckets.
\end{definition}

However, we prove in Theorem~\ref{the:np-hard} that determining whether there exists an acceptable packing of size $n$ is NP-hard, and therefore, finding the optimal partitioned packing is also NP-hard.

\begin{theorem}\label{the:np-hard}

Given a set of tasks $T$, it is NP-hard in $|T|$ to determine whether there exists an acceptable partitioning of $n$ buckets.
We call $\mathcal{P}(T,n)$ this problem.

\end{theorem}

\begin{proof}\emph{(sketch)}
  To prove that this problem is NP-hard, it is
  enough to demonstrate that a certain problem $\mathcal{P}^+$ known
  to be NP-complete can be polynomially reduced to $\mathcal{P}$.

  We recall that the Partition Problem is NP-complete (see~\cite{karmarkar1982difierencing}):
  $\mathcal{P}^+(S)$: given a multiset $S$ of $N$ positive integers $n_i, i \in [0,N-1]$,
  decide whether this multiset can be divided into two submultisets $S_1$ and $S_2$
  such that the sum of the numbers in $S_1$ equals the sum of the
  numbers in $S_2$, and the union of $S_1$ and $S_2$ is included in $S$.

  Let us consider a multiset $S$ and the problem $\mathcal{P}^+(S)$. We assume
  the existence of a deterministic algorithm $A$ that solves
  $\mathcal{P}(T,n)$ in a polynomial time in $|T|$.
  We first distinguish a trivial case where the problem $\mathcal{P}^+(S)$ can be solved in polynomial time.
  Then, we build an algorithm that uses $\mathcal{P}(T,3)$ to solve $\mathcal{P}^+(S)$ in
  polynomial time similarly to the remaining cases, which leads to a contradiction.

  We first consider a trivial case: if there exists $n_k$ in $S$ such
  that $n_k > \sum_{i\in[0,N-1], i \neq k} n_i$, then we return $False$.
  Deciding whether $S$ falls in that specific case is linear in $|S|$, and so is the computation of the answer. %
  If not, let $E$ be a one dimensional space, with bounds $[0, |S|+1[$.
  We build $T$ as a set of $|S|+1$ tasks ($T = \{t_i, i\in [0, N]\}$), such that no task intersects with each other: therefore, the minimum size $m_{T}$ of $T$ (from Definition~\ref{def:acceptable}) will be the same as the size of the biggest task $t$ in $T$.
  The $|S|$ first tasks are all associated with an element of $S$, while the last one will be used to fix $m_{T}$.
  More precisely, we build $T$ as follows:
  \begin{itemize}
  \item the range of $t_i$ is $[i, i+1[$
  \item for $i \neq N$, the weight of $t_i$ is equal to the value of $n_i$ ; $w_{t_N} = \frac{\sum_{i\in[0,N-1]} n_i}{2}$.
  \end{itemize}
  Building $T$ and $t_{max}$ is subpolynomial.

  By hypothesis, $\forall k, n_k \leq \sum_{i \in [0,N-1], i \neq k} n_i$ (as we dealt with this case previously), and by construction, no hyper-rectangle intersects any other, so the minimal weight is the size of the biggest task, which is the last one: $m_T = max_{t_i}(w_{t_i}) = w_{t_{max}}$.
  Therefore, if $S$ can be divided into two submultisets $S_1$ and $S_2$ of the same size, this size is $\frac{\sum_{i\in[0,N-1]} n_i}{2}$, and $\mathcal{P}(T,3)$ answers $True$.

  Reciprocally, if $\mathcal{P}(T,3)$ answers $True$, this means that there exists a packing of size $3$ such that no packing is bigger than $m_T = \frac{\sum_{i\in[0,N-1]} n_i}{2}$.
  In particular, as $w_{t_N} = m_T$, this means that no task is added to the bucket containing it, and that the two remaining buckets contain all tasks $t_i, i\neq N$.
  If one of these buckets where smaller than $m_T$, the other would be bigger than $m_T$ (as $m_T = \frac{\sum_{i\in[0,N-1]} n_i}{2}$), and therefore, both buckets weight exactly $m_T$.
  Therefore, it is possible to separate $S$ in $S_1$ and $S_2$ such that the sum of the numbers in $S_1$ equals the sum of the $S_2$ by taking all the elements corresponding to the tasks in the first bucket for $S_1$, and the elements corresponding to the second bucket for $S_2$.

  Therefore, if we are not in the trivial case treated above, $\mathcal{P}(T,3)$ answers $True$ if and only if $\mathcal{P}^+(S)$ is true in polynomial time.
  As both deciding whether we are in that trivial case and computing the answer in that trivial case can be computed in polynomial time, an algorithm deciding $\mathcal{P}^+(S)$ in polynomial time can be built.
  The assumption of $\mathcal{P}(T,n)$ not being NP-hard leads to a
  polynomial algorithm solving $\mathcal{P}^+$, which is absurd, so
  $\mathcal{P}(T,n)$ is NP-hard.



\end{proof}

\subsubsection{Static packings}
Another point can be highlighted: the difference between what we call \emph{static} packing and \emph{dynamic} partitioning.
Indeed, when trying to optimize the use of partitioned buckets, two main approaches can be used: adapt buckets to tasks, or adapt tasks to buckets.
In the first case, we consider a fixed set of tasks, and try to build partitions in order to minimize the cost of PIR. On the one hand, this optimization makes it possible to perform the best with any set of tasks. On the other hand, as we consider a fixed set of tasks, we may have to compute a new partitioning when this set evolves (when a task is added or removed, at least when it affects the largest bucket).
In the second case however, we consider a fixed partitioning, that is independent from the set of tasks.
This method is more likely to be suboptimal, but it avoids heavy computation of optimal packing and allows a greater flexibility in the context of crowdsourcing, by allowing a large variety of choices and policies from the platform, which can even lead to other kinds of optimization.
For instance, it allows the platform to manage prices policies (\emph{e.g.} making tasks pay for each targeted subspace, higher prices for tasks willing to target highly demanded subspaces, etc.), in order to even the load within the whole space, and to reduce the redundancy of tasks within the PIR library (tasks that target more than one partition).


As finding the optimal is NP-hard, we prefer to set aside dynamic packings, as its main asset is the theoretical possibility to reach optimality while remaining unrealistic in a real-life scenario, and focus instead on static packings.



Static packing means that the design of partitions is independent of
tasks: the tasks contained within the bucket may change, but not the
subspace delimited by the partition. These heuristic packing schemes
are not optimal in general but may be affordable in real-life
scenarios. We propose to use a simple heuristic static packing scheme, the \texttt{PKD PIR Packing},
consisting in using the partitioned space of workers profiles computed
primarily for task design purposes: to each leaf partition
corresponds a bucket containing all the tasks that have metadata
intersecting with it (possibly with padding).
The resulting algorithm is presented in Algorithm~\ref{algo:heur_pir}.
The accordance of this scheme with the distribution of workers can lead to both useful and
efficient buckets (as assessed experimentally, see
Section~\ref{sec:expe}), and the stability over time of the space
partitioning (static approach) makes it easier to design policies to
approach optimality through incentives on the task design (rather than
through the bucket design).

\begin{algorithm}[tbhp]
  \KwData{\\
	  $T$ a Tree computed with the \texttt{PKD} algorithm \\
	  $\mathcal{T}$ a list of tasks\\
  }%
  \KwResult { %
    $P$: A static partitioned packing depending on workers distribution
  }%

  Create an empty packing $P$\\
  \For{all leaves $l$ of the tree $T$}{
	  Create a new empty bucket $b$, assigned to the subvolume of $l$\\
	  \For{all tasks $t$ in $\mathcal{T}$}{
		  \If{$t$ and $l$ intersect}{
		  Add $t$ to the bucket $b$
	  }
	  Add $b$ to $P$
	  }
  }

  \Return{Packing $P$}

  \caption{\texttt{PKD PIR Packing}}%
  \label{algo:heur_pir}
\end{algorithm}

\section{Experimental Validation}%
\label{sec:expe}%
We performed a thorough experimental evaluation of the quality
and performances of both the \texttt{PKD} algorithm and our \texttt{PKD PIR Packing} heuristic (that we abbreviate as \emph{PIR} in the experiments).

\subsection{Datasets}

In this section, we introduce the datasets and data generators that are used in our experiments.

\paragraph{Realistic Dataset.}
To the best of our knowledge there does not exist any reference
dataset of worker profiles that we could use for our
experiments. This led us to building our own dataset from public open
data. The \emph{StackExchange}\footnote{\emph{StackExchange} is a set
  of online forums where users post questions and answers, and vote
  for good answers \url{https://archive.org/download/stackexchange}.}
data dumps are well-known in works related to experts finding. We
decided to use them as well in order to perform experiments on
realistic skills profiles. We computed profiles by extracting skills
from users' posts and votes. In \emph{StackExchange}, users submit
posts (questions or answers) that are tagged with descriptive keywords
(\emph{e.g.,} ``python programming'') and vote positively
(resp. negatively) for the good (resp. bad) answers. We consider then
that each user is a worker, that each tag is a skill, and that the
level of expertise of a given user on a given skill is reflected on
the votes. We favored a simple approach for computing the expertise of
users. First, for each post, we compute a \emph{popularity ratio} as
follows: $r=\mathtt{upvotes}/(\mathtt{upvotes+downvotes})$, where
\texttt{upvotes} is the number of positive votes of the post and
\texttt{downvotes} is the number of negative votes. Second, for each
user $p_i$, for each tag $j$, the aggregate level of expertise
$p_i[j]$ is simply the average popularity ratio of the posts from $p$
tagged by $j$. Note that more elaborate approaches can be used (see
the survey~\cite{Srba:2016:CSC:2988335.2934687}). Finally, we removed
the workers that do not have any skill level higher than $0$. We
applied this method on three \emph{StackExchange} datasets:
\texttt{stackoverflow.com-Posts.7z},
\texttt{stackoverflow.com-Tags.7z}, and
\texttt{stackoverflow.com-Votes.7z} which resulted in 1.3M worker
profiles\footnote{The scripts for generating our dataset are available
  online:
  \url{https://gitlab.inria.fr/crowdguard-public/data/workers-stackoverflow}}.
Figure~\ref{fig:hmwz} (a) shows for ten common skills\footnote{The ten
  common skills considered are the following: \texttt{.net},
  \texttt{html}, \texttt{javascript}, \texttt{css}, \texttt{php},
  \texttt{c}, \texttt{c\#}, \texttt{c++}, \texttt{ruby},
  \texttt{lisp}.} and for the possible levels divided in ten ranges
(\emph{i.e.,} $[0.0, 0.1[$, $[0.1,0.2[$, \ldots, $[0.9,1]$) their
corresponding frequencies. It shows essentially that whatever the
skill considered, most workers have a skill level at 0. The rest of
the distribution is not visible on this graph so we show in
Figure~\ref{fig:hmwz} (b) the same graph but \emph{excluding}, for
each tag, the workers having a skill level at
0. 

\begin{figure*}
  \begin{multicols}{2}[]
    \fbox{\begin{minipage}{1\linewidth} \centering
        \includegraphics[width=1\linewidth]{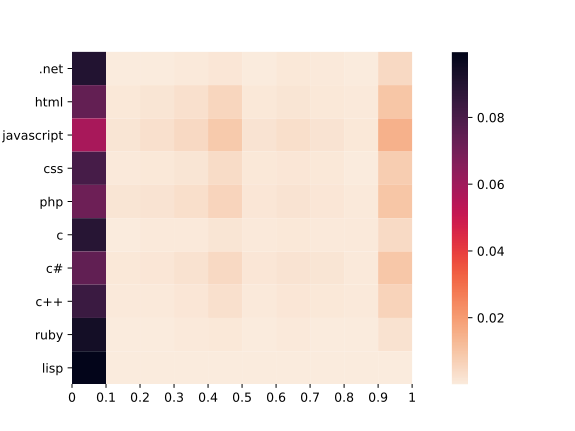} %
        a) On the x-axis: Skill level. On the y-axis: Skill. On the
        heatmap: frequency of the given skill-level on the given
        skill within workers - \emph{including} the workers having a
        skill level at 0. 
      \end{minipage}}

    \columnbreak \fbox{\begin{minipage}{1\linewidth} \centering
        \includegraphics[width=1\linewidth]{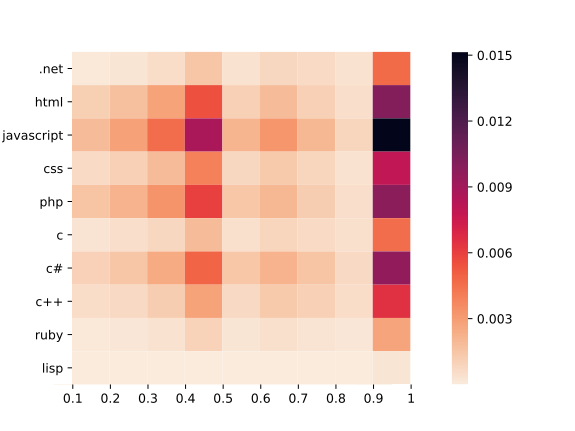} %
        b) On the x-axis: Skill level. On the y-axis: Skill. On the
        heatmap: frequency of the given skill-level on the given
        skill within workers - \emph{excluding} the workers having a
        skill level at 0. 
      \end{minipage}}
  \end{multicols}
  \vspace{-0.5cm}
  \caption{Frequencies of ten common skills within the \texttt{STACK}
    dataset.}
	\vspace{-0.5cm}
  \label{fig:hmwz}
\end{figure*}

\paragraph{Data Generators.} %
We performed our experiments over both synthetic and realistic data.
Our two synthetic generators are specifically dedicated to evaluating
the \texttt{PKD} algorithm with two different kinds of
assumptions. First, our \texttt{UNIF} synthetic data generator draws
skills uniformly at random between $0$ and $1$ (included) (1) for each
dimension of a worker's profile and (2) for each dimension of a task
(more precisely, a min value and a max value per dimension). Second,
our \texttt{ONESPE} generator considers that workers are skilled over
a single dimension and that tasks look for workers over a single
dimension.
The specialty of each worker is chosen uniformly at random, and its
value is drawn uniformly at random between $0.5$ and $1$. The other
skills are drawn uniformly at random between $0$ and $0.5$. Similarly
to workers, the specialty looked for by a task is chosen uniformly at
random as well, its min value is chosen uniformly at random between
$0.5$ and $1$ and its max value is set to $1$. The min values of the other
dimensions of a task are $0$, and their max values are chosen
uniformly at random between $0$ and $0.5$.
Although this second heuristic is obviously not perfect, it seems far more realistic than the previous one.
For the two task generation heuristics, we require that all tasks must contain at least one worker so that the mean error can be correctly computed.

Finally, our realistic data generator, called \texttt{STACK}, consists in
sampling randomly workers (by default with a uniform probability) from
the \texttt{STACK} dataset. %
For our experiments, we generated through \texttt{STACK} workers
uniformly at random and performed the \texttt{ONESPE} task generation
strategy described above. %

\subsection{\texttt{PKD} algorithm}%
\label{sec:quality}

\paragraph{Quality of the \texttt{PKD} algorithm.}

For our experiments, we implemented the \texttt{PKD} algorithm in Python~3 and run our
experimental evaluation on commodity hardware (Linux OS, 8GB RAM, dual
core 2.4GHz). In our experiments, each measure is performed $5$ times
(the bars in our graphs stand for confidence interval), $1k$ tasks, $10$ dimensions, and $\tau = 1$. %

In Fig.~\ref{fig:qualite_1} (a), we fix the privacy budget to
$\epsilon=0.1$, the number of bins to $10$, and the number of workers to $10k$, and we study the impact
of the depth of the tree on the quality.
\texttt{UNIF} achieves the lowest error, as long as the tree is not too deep. This can be explained by the uniform distribution used in the generation method, which matches the uniform assumption within leaves in the tree. When the depth (and the number of leaves) grows, this assumption matters less and less.
\texttt{ONESPE} is more challenging for the \texttt{PKD} algorithm because it is biased towards a single skill. It achieves a higher error but seems to benefit from deeper trees. Indeed, deep trees may be helpful in spotting more accurately the specialized worker targeted.
The results for \texttt{STACK} are very similar.
For all of these distributions, we can see that having a tree deeper than the number of dimensions leads to a significant loss in quality.

In Fig.~\ref{fig:qualite_1} (b), we analyze the variations of quality according to the value of $\epsilon$, with $10$ bins, a depth of $10$, and $10k$ workers.
In this case, the relative error seems to converge to a non-zero minimum when $\epsilon$ grows, probably due to inherent limits of KD-Tree's precision for tasks.

In Fig.~\ref{fig:qualite_1} (c), we fix the privacy budget to $\epsilon=0.1$, the depth of the tree to $10$ and $10k$ workers. We can see the impact of the number of
bins for each histogram used to compute a secure median. This value
does not greatly impact the relative error for the \texttt{UNIF} and \texttt{STACK} models, although we can see that performing with $1$ bin seems to give slightly less interesting results, as it looses its adaptability toward distributions.
For the \texttt{ONESPE} model, having only $1$ bin gives better results: indeed, the uniformity assumption within the bin implies that all dimensions are cut at $0.5$, which is also by construction the most important value to classify workers generated with this procedure.

In Fig.~\ref{fig:qualite_1} (d), we compare the quality according to the number of workers with $\epsilon=0.1$, $10$ bins and a depths of $10$. As the $\epsilon$ budget is the same, the noise is independent from this number, and thus, the quality increases with the number of workers.

We can notice that our results for the relative error are quite close to the state of the art results, such as the experiments from~\cite{psd}, which are performed on $2$-dimensional spaces only, with strong restrictions on the shapes of queries (tasks in our context) and in a centralized context.

\begin{figure}
  \centering
\begin{multicols}{2}[]
      \includegraphics[width=1.1\linewidth]{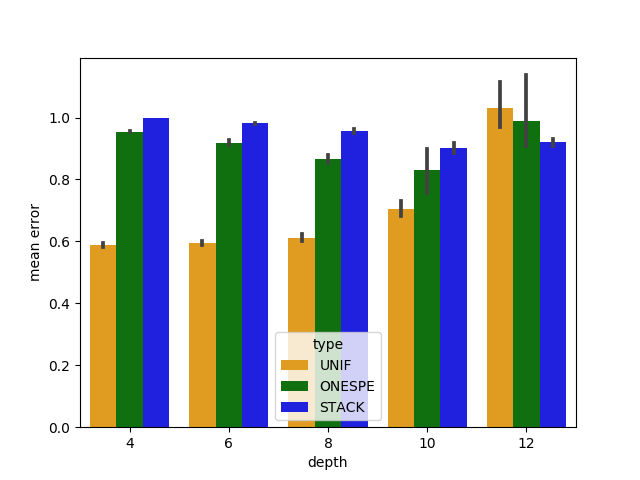}
      a) Variations according to the depth of the tree.\\
	  $10$ dimensions, $10k$ workers, $1k$ tasks, $\tau = 1$, $\epsilon = 0.1$, $10$ bins

      \includegraphics[width=1.1\linewidth]{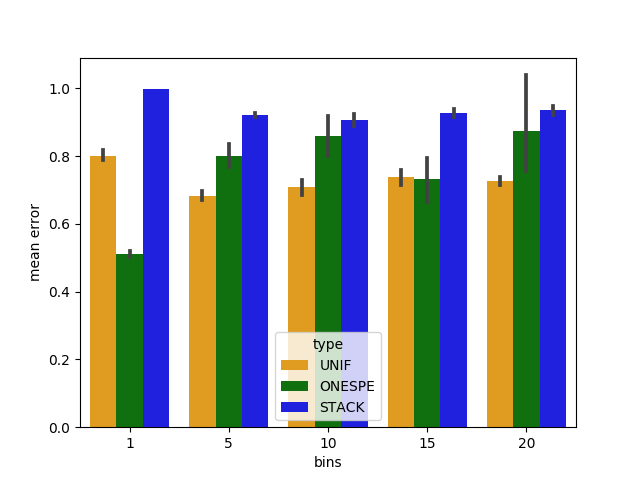}
      c) Variations according to the number of bins.\\
	  $10$ dimensions, $10k$ workers, $1k$ tasks, $\tau = 1$, $\epsilon = 0.1$, $depth = 10$

  \columnbreak
      \includegraphics[width=1.1\linewidth]{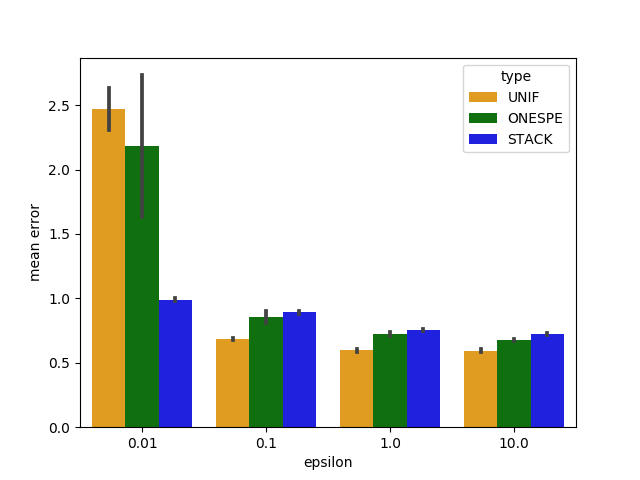}
      b) Variations according to $\epsilon$ privacy budget.\\
	  $10$ dimensions, $10k$ workers, $1k$ tasks, $\tau = 1$, $10$ bins, $depth = 10$

	  \includegraphics[width=1.1\linewidth]{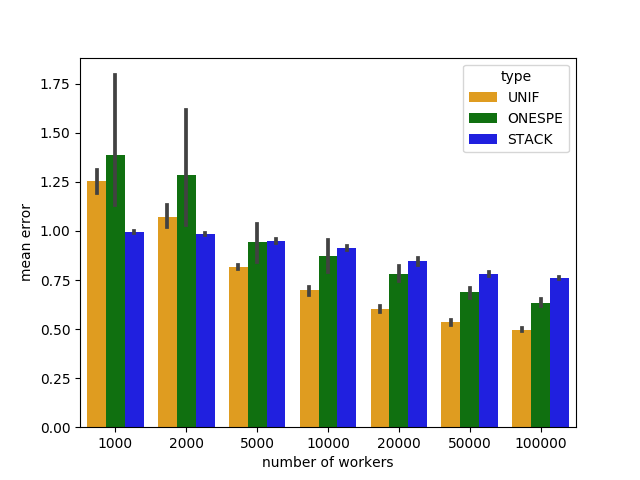}
      d) Variations according to the number of workers.\\
	  $10$ dimensions, $1k$ tasks, $\tau = 1$, $\epsilon = 0.1$, $10$ bins, $depth = 10$
\end{multicols}
\vspace{-0.5cm}

\caption{Mean relative error (see Definition~\ref{eq:quality}, the lower the better)}
\label{fig:qualite_1}
\vspace{-0.5cm}
\end{figure}

\paragraph{Computation time of the \texttt{PKD} algorithm.}

Our performance experiments were performed on a laptop running Linux
OS, equipped with $16GB$ of RAM and an Intel Core $i7-7600U$
processor. We measured the average computation time across $100$
experiments of each of the atomic operations used in the \texttt{PKD}
algorithm: encryption, partial decryption, and encrypted addition. The
results are summed up in Fig.~\ref{fig:crypto-comput}, with keys of
size $2048$ bits, using the University of Texas at Dallas
implementation for its accessibility%
\footnote{\url{http://cs.utdallas.edu/dspl/cgi-bin/pailliertoolbox/index.php?go=download}}. %
We use our cost analysis together with these atomic measures for
estimating the global cost of the \texttt{PKD} algorithm over large
populations of workers (see Equation~\ref{eq:w-send-total},
Equation~\ref{eq:w-send-avg}, and Equation~\ref{eq:pf-send} in
Section~\ref{sec:compl-analysis}).

\begin{figure}
	\centering
  \includegraphics[width=.8\linewidth]{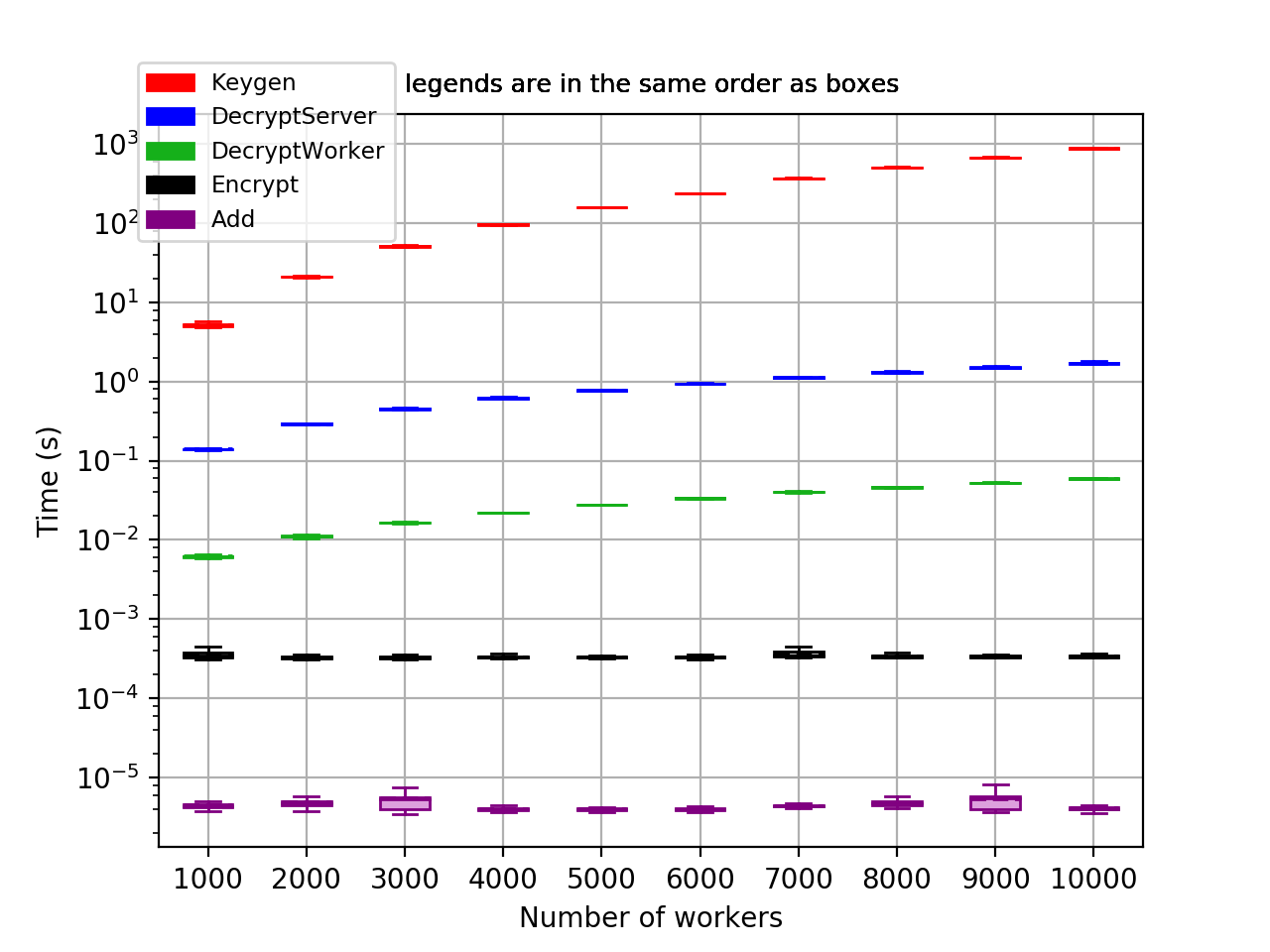}
  \vspace{-0.5cm}
  \caption{Computation time of homomorphically encrypted operations}
	\vspace{-0.5cm}
  \label{fig:crypto-comput}
\end{figure}

We can observe that the slowest operation is by far the generation of
the keys. However, since this operation is performed only once, the
cost of less than $1000$ seconds (about $17$ minutes) for $10k$
workers is very reasonable: this operation can be performed as soon as
there are enough subscriptions, and the keys may be distributed
whenever the workers connect. The other operations are faster
individually, but they are also performed more often. For $10k$
workers, $10$ workers required for decryption, a depth of the KD-Tree
of $10$ and $10$ bins, we can observe that: each worker will spend
less than $10$ seconds performing
encryptions, 
the platform will spend less than $1000$ seconds performing encrypted
additions, the average worker will spend less than $1$ second
performing decryptions, and the platform will spend less than $3000$
seconds performing decryptions.

Overall, these costs are quite light on the worker side: less than
$20$ seconds with commodity hardware. On the server side, the
computation is more expensive (about one hour), but we could expect a
server to run on a machine more powerful than the one we used in our
experiments. Additionally, it is worth to note that: (1) the
perturbed skills distribution is computed only once for a given a
population of workers and then used repeatedly, and (2) we do not have
any real time constraints so that the \texttt{PKD} algorithm can run
in background in an opportunistic manner.

\subsection{Assignment using packing}

\paragraph{Quality of our packing.}

We here propose to evaluate the quality of our partitioned packing approach.
Our experiments are performed with the same settings as those used to measure the quality of the \texttt{PKD} algorithm (see Section~\ref{sec:quality}).
To do so, we propose two main metrics.
First, me measure the mean precision for tasks, as defined in Definition~\ref{def:precision}.
Although this measure is useful to understand the overall improvement of our approach, it does not take into account the fact that downloads caused by PIR scale with the largest item.
Therefore, we introduce a second measure, the mean number of tasks that a worker would download. This value, that we call \emph{maximum tasks}, is computed as the maximum number of tasks that a leaf of the KD-tree intersects with: indeed, due to Condition~\ref{def:packing}.\ref{enum:packcondsize} (PIR requirement), all workers will download as many data as contained in the biggest bucket.

\begin{figure}
  \centering
  \vspace{-0.5cm}
\begin{multicols}{2}[]
      \includegraphics[width=1.1\linewidth]{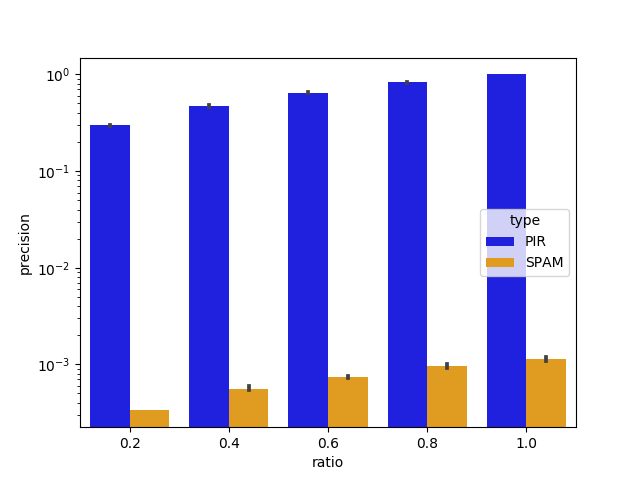}
	  a) Precision in log scale, according to the ratio of leaf taken by task for the \texttt{UNIF} model.\\
	  $10$ dimensions, $10k$ workers, $1k$ tasks, $\tau = 1$, $\epsilon = 0.1$, $10$ bins, $depth = 10$

  \columnbreak
      \includegraphics[width=1.1\linewidth]{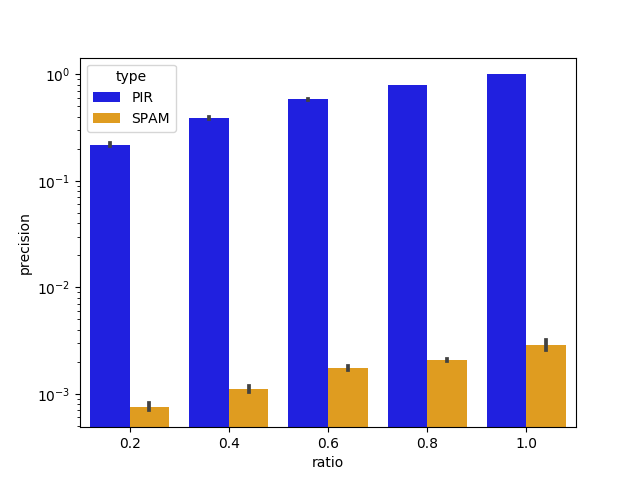}
	  b) Precision in log scale, according to the ratio of leaf taken by task for the \texttt{ONESPE} model.\\
	  $10$ dimensions, $10k$ workers, $1k$ tasks, $\tau = 1$, $\epsilon = 0.1$, $10$ bins, $depth = 10$
\end{multicols}
\vspace{-0.5cm}
\caption{Precision (the higher the better)}
	\vspace{-0.5cm}
\label{fig:precision_pourcent}
\end{figure}

In the task generation methods introduced previously, tasks are built independently from the KD-tree itself. This independence was logical to measure the quality of the \texttt{PKD} algorithm.
However, this very independence leads to poor results when it comes to building efficient packing on top of a KD-tree: as tasks are independent from the KD-tree, they have little restriction on how small they are (meaning that few workers will match with them, although all workers in leaf that intersect with it will download it), or on how many leaves they intersect with, leading to low precision, and high size of buckets.

Therefore, we introduce a new method to build tasks: \texttt{SUBVOLUME}. With this method, we build tasks as \emph{subleaves}, meaning that all tasks are strictly included within \emph{one} leaf of the KD-tree. Furthermore, we also enforce the size of the task as a parameter, such that the volume of the task is equal to a given ratio of the task.
More precisely, for a ratio $r \in [0,1]$, a space $E$ of $d$ dimensions and a picked leaf $l$, the interval of a task in a given dimension $d_i$ $l_{d_i} \times r^{1/d}$, where $l_{d_i}$ is the interval of the leaf in dimension $d_i$.
The \texttt{SUBVOLUME} model of tasks can easily be introduced by economic incentives from the platform, such as having requesters pay for each targeted leaf, which is likely to induce a maximization of the volume taken, and a reduction of the tasks that intersect with more than one leaf.
Note that we do not perform experiments with this generation of tasks on the \texttt{Stack} dataset, as most workers have their skills set to either $0$ or $1$, which leads to very unreliable results as tasks almost never encompass either of these values.

The comparison between the \texttt{PKD PIR Packing} heuristic and the spamming approach using this new method to generate tasks, presented in Figure~\ref{fig:precision_pourcent}, show that our approach improves precision by at least two orders of magnitude.
Also, note that for $r = 1$, the precision is equal to $1$ in the PIR approach. This result comes from the fact that, with $r = 1$, all workers within a leaf are targeted by all tasks that intersect with that leaf, meaning that they do not download irrelevant tasks.

\begin{figure}
  \centering

  \vspace{-1cm}
  \includegraphics[width=.7\linewidth]{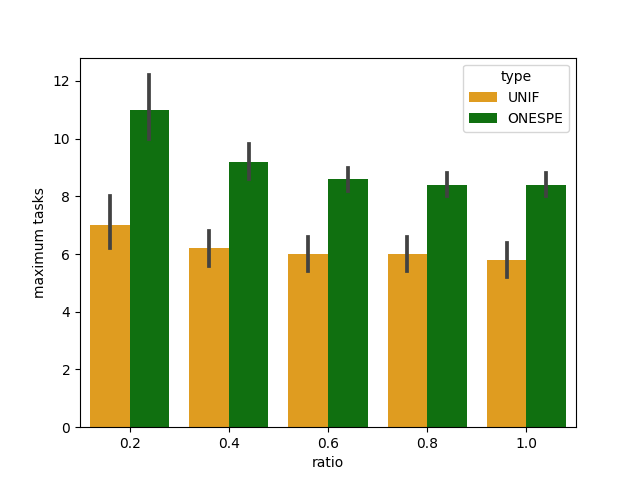}

\vspace{-0.4cm}
\caption{Number of tasks downloaded according to the ratio of leaf taken by task for the packing approach.
$10$ dimensions, $10k$ workers, $1k$ tasks, $\tau = 1$, $\epsilon = 0.1$, $10$ bins, $depth = 10$}
	\vspace{-0.5cm}
\label{fig:expe_maxleaf}
\end{figure}

The maximum number of tasks connected to a leaf, showed in Figure~\ref{fig:expe_maxleaf}, show that the cost of download is also significantly improved (these values are to be compared to $1000$, the total number of tasks that are downloaded with the spamming approach) also shows great improvement (around $2$ orders of magnitude), as tasks are more evenly spread within the leaves (there are $2^{10} = 1024$ leaves for a depth $10$ of the tree, which can explain this improvement).

\paragraph{Cost of the PIR protocol.}

We here study the impact of the number of files and of the size of files on the computation time.
In the experiments, we used a computer with $8GB$ of RAM, and a Ryzen $5$ $1700$ processor, using the implementation of~\cite{xpir}\footnote{\url{https://github.com/XPIR-team/XPIR}}.

As we can see in Figure~\ref{fig:pir} with keys of size $1024$ bits,
computation time is proportional to the overall size of the PIR
library (the coefficient of determination gives $r^2=0.9963$), and
that it grows at $0.14s/MB$ for a given request, as long as the
library can be stored in RAM\footnote{If it cannot, accesses to the secondary storage device are necessary. This would increase the runtime accordingly. However, since the library is scanned once per query, sequentially, the cost would remain linear in the size of the library.}.

\begin{figure}
	\centering
  \includegraphics[width=0.7\linewidth]{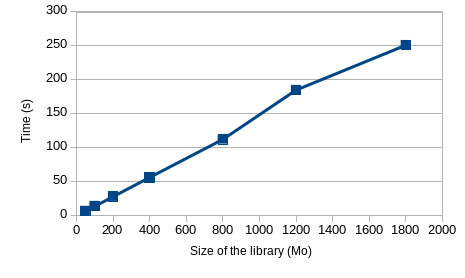}
  \vspace{-0.5cm}
  \caption{Computation time of the retrieval of a file according to the PIR library's size}
	\vspace{-0.5cm}
  \label{fig:pir}
\end{figure}

We now evaluate the maximum number of tasks $n_{max}$ that our system can take into account, according to two parameters: first, the time $t$ that workers accept to wait before the download begins, and second, the size $s$ that workers accept to download.
As $n_{max}$ does not solely depend on $t$ and $s$, we introduce a few other notations:
\begin{itemize}
	\item $f$ is the expansion factor of the encryption scheme.
	\item $|task|$ the mean size of a task.
	\item $k$ the proportion of tasks that are in the biggest leaf of the KD-tree (for instance, $k = 0.1$ means that the biggest leaf contains one tenth of all tasks)
	\item $depth$, the depth of the KD-tree (that is linked with the number of buckets)
\end{itemize}

In the spamming approach, the maximum number of tasks that can be managed by our system is independent from $t$ and can be simply computed as: $$n_{max, SPAM} = \frac{s}{|task|}$$
For the \texttt{PKD PIR Packing} heuristic, both $s$ and $t$ lead to a limitation on $n_{max,PIR}$.
We first consider the limit on the computation time $t$: according to our results in Figure~\ref{fig:pir}, the PIR library cannot be bigger than $\frac{t}{0.14}$, and the size of a bucket, can be computed as $k \times |task| \times n_{max,PIR}$ (by definition of $k$, as all buckets weight as much as the biggest one).
As the library can be computed as the product of the number of buckets and their size, this leads us to $2^{depth} \times k \times |task| \times n_{max,PIR} \leq \frac{t}{0.14}$, or equivalently $n_{max,PIR} \leq \frac{t}{0.14 \times 2^{depth} \times k \times |task|}$.
We now consider the limit $s$ on the size of download. For each worker, the size of a download will be the same, computed as the product of the expansion factor and the size of a bucket: $f \times k \times n_{max,PIR} \times |task| \leq s$.
This inequality leads to $n_{max,PIR} \leq \frac{s}{f \times |task| \times k}$.
By combining these two inequalities, $n_{max,PIR}$ takes its maximum value when
$$n_{max,PIR} = min(\frac{s}{f \times |task| \times k}, \frac{t}{2^{depth} \times 0.14 \times |task| \times k})$$

In Figure~\ref{fig:cost_pir}, we compare the number of tasks that a crowdsourcing platform can manage with different values of $t$ and $s$, using either the spamming approach or our \texttt{PKD PIR Packing} heuristic.
For the sake of simplicity, we consider that the expansion factor $f$ is $10$, although smaller values are reachable with XPIR protocol~\cite{xpir}. This factor will impact the amount of tasks that a worker can download.
We take $d=10$ similarly to our previous experiments.
We consider a mean size of task $|task| = 1MB$. It can be noticed that $|task|$ has no impact on the comparison ($\frac{max_{n,PIR}}{max_{n,SPAM}}$ does not depend on $|task|$).
For $k$, we consider two possible values: $k=0.01$, as suggested by the experiments in Figure~\ref{fig:expe_maxleaf}, and $k=\frac{1}{2^{10}}$, which represents the optimal case, where tasks are perfectly spread among buckets (for instance, due to strong incentives from the platform).

\begin{figure}
  \centering

\begin{multicols}{2}[]
      \includegraphics[width=1.1\linewidth]{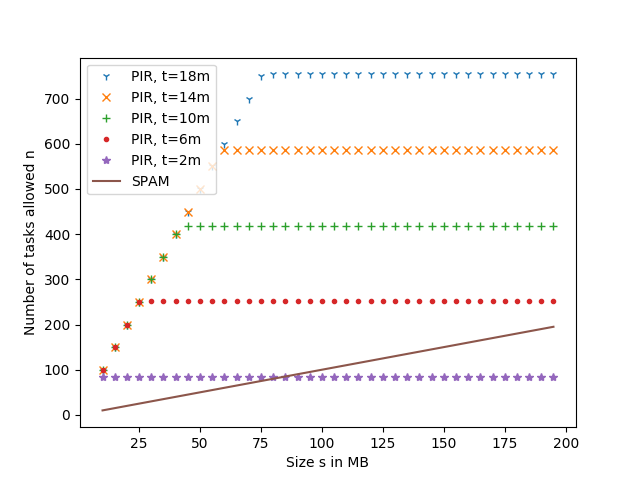}
	  a) Number of tasks $n$ manageable by our system according to the size $s$ a worker accepts to download. \\
	  $k=0.01$, $|task|=1MB$, $f=10$, $depth=10$

  \columnbreak
      \includegraphics[width=1.1\linewidth]{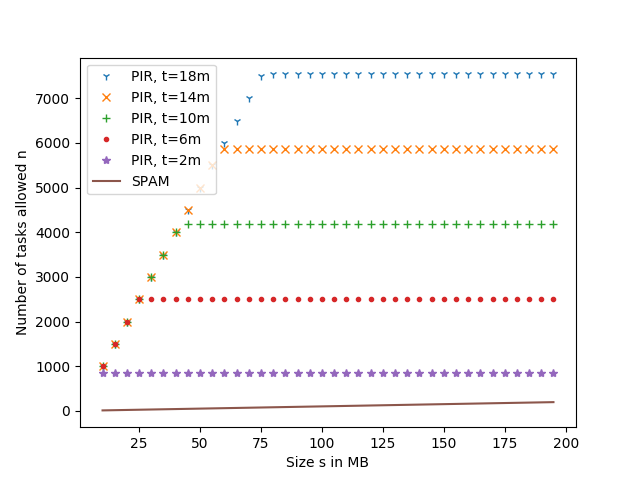}
	  b) Number of tasks $n$ manageable by our system according to the size $s$ a worker accepts to download. \\
	  $k=0.001$, $|task|=1MB$, $f=10$, $depth=10$
\end{multicols}
\vspace{-0.5cm}
\caption{Precision (the higher the better) ; curves are in the same order as the captions}
	\vspace{-0.5cm}
\label{fig:cost_pir}
\end{figure}

In these experiments, we can notice that our approach depends on both the computation time allowed and the size of the number of task in the largest bucket. In a real-life scenario, platforms would benefit from enforcing incentives to even the load between buckets.
However, if workers are willing to limit their download to less than $100MB$, the \texttt{PKD PIR Packing} heuristic outperforms the spamming approach as long as users are willing to limit their download even with relatively short computation times (less than $10$ minutes) by up to several orders of magnitude.
Our method is especially interesting in settings where the bandwidth is low (\emph{e.g.} with mobile devices), with low values of $s$.
On the opposite, it is interesting to highlight that high computation times are not necessarily prohibitive: as the computation is performed by the platform, a worker could very well ask for a bucket of tasks and download it later on when it is ready.

\section{Discussion}
\label{sec:discussion}
In this section, we propose a discussion on questions raised by our work that are not our primary focus.
More precisely, we elaborate our views on updates that our system may or may not allow (both for the \texttt{PKD} algorithm and the \texttt{PKD PIR Packing} heuristic), with some advantages and drawbacks.

\subsection{Updating tasks and PIR libraries}
\label{sec:reducing_costs}

In this work, we dealt with the download of tasks as a \emph{one-shot} download, meaning that a worker will download tasks once and for all.
However, in a real-life scenario tasks are likely to evolve (\emph{e.g.} new tasks will be added and old tasks will be outdated), and workers are equally likely to update their tasks.
Without further improvement, our design would require each worker to download a whole packing for each update of the available tasks.
However, more elaborate approaches are possible. Although it is not our focus to develop them exhaustively, we propose a few tracks that are likely to diminish the costs greatly.

For that purpose, we propose to divide time into fixed duration \emph{periods} (\emph{e.g.} a day, a week, etc.) and to
additionally take into account the period at which a task is issued in
order to pack it.
We give below two options for allowing updates.
Although their improvement have not been quantified nor validated experimentally,
These schemes aim at increasing the memory cost on the server in order to alleviate the overall computation required.


\subsubsection{Packing by Period}%
A simple scheme that allows easier updates while reducing the size of
single PIR request consists in designing packing not only according to
a specific partitioning but also according to time periods.
The platform builds one PIR library per period, \emph{i.e.,} considering
only the tasks received during that period.\footnote{In this kind of
  methods, a task can be either maintained into its starting period up till
  it's lifespan, or one can consider keeping up a limited number of
  periods (\emph{e.g.} all daily periods for the current month) and re-adding tasks on new periods packing each time they
  are deleted (\emph{e.g.} for tasks that are meant to be longer than a month).
  More elaborate or intermediate methods are also possible, but we will not explore this compromise in this paper.}.
  Workers simply need to perform PIR requests over the missing period(s) (one request per missing
period). As a result, the \texttt{PIR-get} function is executed on the
library of the requested period, which is smaller than or equal to the
initial library.

However, this scheme may result in high costs if the distribution of tasks is skewed.
For instance, let's consider two time periods $p_1$ and $p_2$, two subspaces of the space of skills $s_1$ and $s_2$, and three tasks $t_1$, $t_2$ and $t_3$ such that $t_1$ and $t_2$ appear only in $p_1$ and $s_1$, while $t_3$ appears only in $p_2$ and $s_2$.
In that case, all workers will download first the PIR item for period $p_1$, which is the same size as $w_{t_1}+w_{t_2}$ (due to padding for workers not in $p_1$) and then a second PIR item for $p_2$, of size $w_{t_3}$.
Without that period strategy, a worker who performs regular updates would have downloaded tasks $t_1$ and $t_2$ (or equivalent size) twice due to the update, and $t_3$ once, but a worker who would not have performed the intermediary download would have downloaded $max(w_{t_3}, w_{t_1}+w_{t_2})$.
Therefore workers who update frequently would benefit from this strategy, while workers who do not would have worse results.

%


\subsubsection{Personalized Packing by Period}
In order to tackle the previously mentioned issue caused by skewed
distribution of tasks, and to optimize the size of the downloaded
bucket for any frequency of downloads, we propose to adapt
the packing to the workers frequency of downloads.

Indeed, we observe that it is enough to perform as many packings as there are possible
time-lapses for workers, \emph{e.g.,} one packing for the last period,
one packing for the last two periods, one packing for the last three
periods, \emph{etc.}.
As a result, each \texttt{PIR-get} request is
associated with a time-lapse in order to let the PIR server compute
the buckets to be downloaded (or use pre-computed buckets).
With this method, we can get the best of both worlds with the previous example:
someone who downloads frequently will only have small updates, while someone who does not will not suffer from overcosts.

The main (and limited) drawback of this method is that the platform will have to store multiple PIR-libraries, which increases the storage required.

\subsubsection{Security of Packing by Period}

In both of the above schemes, we consider multiple downloads from workers. Even worse, in the second case the number of downloads may vary depending on workers habits.
If the above proposition were to be used, more accurate proofs of security would have to be done. Although it is not our focus to propose them in this article, we provide here some intuitions on their requirements.
In the first case, the number of downloads is the same for all workers, and would therefore not lead to great modifications of our proof.
In the second case however, the number of downloads depends on the frequency of downloads of workers. In order not to reveal information about worker's profiles, a new hypothesis is likely to be required, that states or implies that the frequency of downloads of workers is independent from their profiles.

\subsection{Updating \texttt{PKD}}

The \texttt{PKD} algorithm is not meant to allow users to update their profiles, as they would have to communicate information to do it, and this would either break our security policy, or exceed the $\epsilon$ privacy budget.
However, departures or arrivals are not inherently forbidden by our security policy.
A simple and naive way to upgrade the \texttt{PKD} algorithm to take new arrivals into account is to create multiple KD-trees, and to combine them.
For instance, one could imagine using the \texttt{PKD} algorithm on every new $k$ arrivals (\emph{e.g.} $k=1000$ or $k=10 000$). The estimation of workers within a subspace would be the sum of the estimations for each KD-tree, and a new PIR library could be built for each of these KD-trees.
For retrieval of workers, as it is impossible to know where the worker was, the most naive way to proceed is to retrieve a given value to each leaf of the approximated KD-tree, for instance $\frac{n_{leaf}}{n_{tree}}$, where $n_{leaf}$ is the approximated number of workers in the leaf, and $n_{tree}$ the total number of workers.
Once again, more elaborate methods are possible, but stand out of the focus of this paper.

\section{Related Work}
\label{sec:rel}%

\paragraph{Privacy-Preserving Task Assignment.}
Recent works have focused on the privacy-preserving task assignment
problem in general crowdsourcing. In~\cite{beziaud}, each worker
profile - a vector of bits - is perturbed locally by the corresponding
worker, based on a local differentially private bit flipping scheme. A
classical task-assignment algorithm can then be launched on the
perturbed profiles and the tasks.
An alternative approach to privacy-preserving task-assignment
has been proposed in~\cite{kajino2015phd}. It is based on the
extensive use of additively-homomorphic encryption, so it does not
suffer from any information loss, but this has a prohibitive cost in
terms of performance. Other works have focused on the specific context of spatial
crowdsourcing~\cite{to2018privacy,zhai2018towards,to2014framework}.
They essentially differ from the former in that spatial crowdsourcing
focuses on a small number of dimensions (typically, the two dimensions
of a geolocation) and is often incompatible with static worker
profiles. All these works explore solutions to ensure an assignment
between tasks and workers in a private way, and are complementary to
our approach.

\paragraph{Decentralized Privacy-Preserving Crowdsourcing Platform.}
ZebraLancer~\cite{zebralancer} is a decentralized crowdsourcing
platform based on blockchains, zero-knowledge proofs, and smart
contracts and focuses on the integrity of the reward policies and the
privacy of the submissions of workers against malicious workers or
requesters ({\em e.g.,\/} spammers, free-riders). Zebralancer does not
consider using worker profiles (neither primary nor secondary usages).


%

\paragraph{Privacy-Preserving KD-Trees.}
The creation and the publication of private KD-Trees has been studied
in depth in~\cite{psd}, but in our context, this work suffers from two
main defficiencies. First, it considers a trusted third party in
charge of performing all the computations while in our work we do not
assume any trusted third party. Second, it restricts the number of
dimensions to two, which is unrealistic in our high-skills
crowdsourcing context.  Enhancements to the technique have been
proposed, for example~\cite{qardaji2013differentially}, but without
tackling the trusted third party assumption. %
%

\paragraph{Privacy-Preserving COUNTs.}
Other differentially private count algorithms exist and use histograms. With the use of constrained inference, the approaches proposed {\em e.g.,\/} in \cite{Qardaji:2013:UHM:2556549.2556576,hay2010boosting} outperform standard methods. But they are limited to centralized contexts with a trusted third party, and only consider datasets with at most three dimensions. %
The {\tt PrivTree} approach~\cite{Zhang:2016:PDP:2882903.2882928} eliminates the need of fixing the height of trees beforehand, but their security model also considers a trusted third party, and their expriments are limited to four dimensions, which is lower than the number of skills that we consider. %
{\tt DPBench}~\cite{hay2016principled} benchmarks these methods in a centralized context and considers one or two dimensions. %
Finally, the authors of~\cite{kellaris2018engineering} tackle the efficiency issues of privacy-preserving hierarchies of histograms. It suffers from the same dimension and privacy limitations as the above works.

\paragraph{Task Design.}
To the best of our knowledge, the problem of designing a task
 according to the actual crowd while providing sound privacy
 guarantees has not been studied by related works. Most works focus on
 the complexity of the task~\cite{Keepitsimple}, on the interface with
 the
 worker~\cite{li2016crowdsourced,Keepitsimple,kucherbaev2015crowdsourcing},
 on the design of workflows~\cite{kulkarni2011turkomatic,kulkarni2012collaboratively}, or on the filters that may be embedded within tasks and
 based on which relevant workers should be selected~\cite{allahbakhsh2013quality}. However, these approaches
 ignore the relevance of tasks with respect to the actual crowd, and
 thus ignore the related privacy issues.

\section{Conclusion}%
\label{sec:conc}%
We have presented a privacy-preserving approach dedicated to enabling
various usages of worker profiles by the platform or by requesters,
including in particular the design of tasks according to the actual
distribution of skills of a population of workers. We have proposed
the {\tt PKD} algorithm, an algorithm resulting from rethinking the
{\em KD-tree\/} construction algorithm and combining
additively-homomorphic encryption with differentially-private
perturbation. No trusted centralized platform is needed: the {\tt
PKD} algorithm is distributed between workers and the platform. We
have provided formal security proofs and complexity analysis, and an
extensive experimental evaluation over synthetic and realistic data
that shows that the {\tt PKD} algorithm can be used even with a low
privacy budget and with a reasonable number of skills. Exciting future
works especially include considering stronger attack models
(\emph{e.g.,} covert or malicious adversaries), evaluating more precisely our propositions for updates, protecting the tasks in addition
to worker profiles, and guaranteeing the integrity of worker profiles.

\bibliographystyle{splncs04}
\bibliography{biblio}



\end{document}